\documentclass[11pt]{article}

\usepackage[final]{acl}

\usepackage{times}
\usepackage{latexsym}

\usepackage[T1]{fontenc}
\usepackage[utf8]{inputenc}

\usepackage{microtype}
\usepackage{inconsolata}

\usepackage{graphicx}

\usepackage{booktabs}
\usepackage{multirow}
\usepackage{tabularx}
\usepackage{array} 
\newcolumntype{R}{>{\raggedleft\arraybackslash}X}
\newcolumntype{S}{>{\hsize=0.9\hsize\raggedleft\arraybackslash}X}
\newcolumntype{W}{>{\hsize=1.6\hsize\raggedleft\arraybackslash}X}

\newcolumntype{T}{>{\hsize=0.8\hsize\raggedleft\arraybackslash}X}
\newcolumntype{M}{>{\hsize=1.4\hsize\raggedleft\arraybackslash}X}

\usepackage{amsmath, amssymb, amsfonts, amsthm}
\usepackage{xcolor}
\usepackage{mdframed}             %

\usepackage{algorithm}
\usepackage{algpseudocode}        %

\algrenewcommand\algorithmicindent{1.2em}

\usepackage[font=small, labelfont=bf, labelsep=period]{caption}

\title{Characterizing Opinion Evolution of Networked LLMs}

\author{  \textbf{Caleb Probine},
 \textbf{Yigit Ege Bayiz},
 \textbf{Filippos Fotiadis},
 \\
 \textbf{Samuel Li},
 \textbf{Yunhao Yang},
 \textbf{Ufuk Topcu}
 \\
 The University of Texas at Austin
 \\
    \texttt{\{cprobine, egebayiz, ffotiadis, samuel.li, yunhaoyang234, utopcu\}@utexas.edu}
 }

\usepackage{comment}
\usepackage{nicefrac}
\usepackage{placeins}

\begin{document}
\maketitle
\begin{abstract}

Large language models (LLMs) increasingly interact with one another in multi-agent systems, from simulations of human discourse to influence operations and fully LLM-driven social platforms. 
These interactions give rise to new regimes of opinion propagation that are not yet well understood.
We investigate whether classical opinion dynamics models, which have long been used to explain how interactions shape collective beliefs in human societies, can capture the behavior of LLM networks. Specifically, we quantify the relative contributions of classical opinion-dynamics mechanisms, including neighbor influence, stubbornness, homophily, and bias, to the evolution of LLM opinions.
Among the mechanisms considered, a bias term capturing regression toward a model- and topic-dependent opinion emerges as a significant explanatory component, with its inclusion reducing cumulative estimated mean opinion error by up to  $88\%$. We additionally observe similar trends across the evaluated model families, discussion topics, and network structures.

\end{abstract}

\section{Introduction}
\label{sec:intro}

Networked systems of large language models (LLMs) are increasingly deployed across social media and information environments.
They have been studied as tools for simulating opinion formation in purely human settings \cite{chuang2024simulating, tornberg2023simulating, wang2025decoding}, are emerging as instruments of influence operations in mixed human–LLM networks \cite{goldstein2023generative}, and most recently, have led to fully LLM-driven social platforms \cite{lin2026exploring}. These developments raise fundamental questions about how opinions form, evolve, and propagate in networks of interacting LLMs.

Classical opinion dynamics models offer a promising lens for reasoning about the emergent behavior of groups of LLMs. In human social systems, such models have proven powerful for studying how effects such as stubbornness, homophily, and bias drive belief evolution. Extending them to interacting LLMs may therefore provide a principled framework for understanding emergent behavior in LLM networks. When LLMs are used to simulate humans, for instance, opinion dynamics models can certify when LLM populations behave as faithful surrogates for human ones. For influence operations, they can identify the leverage points through which stubborn LLM agents disproportionately shape collective beliefs. For LLM-driven social platforms, they can guide design choices that mitigate echo chambers driven by homophily. Realizing these uses, however, requires knowing which mechanisms, such as stubbornness, homophily, and bias, actually drive LLM networks.

We investigate which mechanisms from classical opinion dynamics models best explain LLMs' belief evolution, and we quantify their relative contributions.
Recent work \cite{abedini2026don,yazici2026opinion,ju2024sense} provides initial insight into the alignment of networked LLMs' opinions with models such as DeGroot \cite{degroot1974reaching}, Friedkin-Johnsen \cite{friedkin1990social}, or Hegselmann-Krause \cite{hegselmann2002opinion} dynamics.
Each of these models hard-codes a specific combination of stubbornness, homophily, and bias.
In contrast, our contribution is to explicitly quantify the relative contributions of stubbornness, homophily, and bias in governing the evolution of LLMs' beliefs.

We study the following problem.
\begin{itemize}
    \item \textbf{RQ1}: Which components of classical opinion dynamics models, such as stubbornness, homophily, and bias, are most critical for explaining LLMs' belief evolution?
\end{itemize}

We also investigate the extent to which conclusions about the importance of these components generalize across different settings.
\begin{itemize}
    \item \textbf{RQ2}: To what extent do the derived opinion dynamics parameters differ between networks of different LLMs and discussion topics? 
    \item \textbf{RQ3}: Are the recovered opinion dynamics a property of the LLM itself, or of the particular network on which it communicates?
\end{itemize}

By answering these questions, we lay the groundwork for the design, control, and validation of opinion dynamics across networks that include LLMs.

To address these questions, we simulate discussions among groups of LLMs, extract their opinions, and fit opinion dynamics models to the resulting trajectories. We first simulate discussions in which LLM agents exchange messages over a communication graph on contested topics. A stance identification model then converts each message into a numerical opinion, yielding opinion trajectories over time. We fit classical opinion dynamics models to these trajectories and quantify how well each model, and each of its components, reproduces the observed behavior, thus addressing RQ1. We then repeat this procedure across the Qwen, Llama, and Gemma families and across topics, testing whether the same components dominate, addressing RQ2. Finally, while RQ1 and RQ2 assume each agent is influenced uniformly by its neighbors, we fit per-edge influence weights and test whether this added flexibility substantially improves prediction accuracy, addressing RQ3.

\subsection*{Related work}

\paragraph{Classical opinion dynamics}
Classical opinion dynamics provides a rich set of models for information flow in social networks.
In the simplest opinion dynamics model, i.e., the DeGroot model \cite{degroot1974reaching}, agents update their opinions by averaging amongst their neighbors.
The Friedkin-Johnsen (FJ) model then reflects stubbornness by modeling an attraction to initial opinions \cite{friedkin1990social}.
Meanwhile, in bounded-confidence models, agents only exchange opinions with their neighbors if their opinions are sufficiently close \cite{deffuant2000mixing,hegselmann2002opinion}.
We aim to explore the extent to which various components of these opinion dynamics models reflect the collective behavior of networks of large language models (LLMs). 

\paragraph{Generative agent-based modeling}
Numerous authors propose the use of LLMs for simulation of human opinion dynamics \cite{chuang2024simulating,tornberg2023simulating,composta2025simulating,wang2025decoding,liu2025beliefs,donkers2025human,yao2025fusing,mou2024unveiling}.
However, rather than proposing LLMs replace classical opinion dynamics models, we investigate whether we can model the evolution of collective beliefs among networked LLMs with classical opinion dynamics.
Some works compare the performance of LLM-based opinion dynamics models to classical models \cite{mou2024unveiling,yao2025fusing,gu2025large,liu2025beliefs,composta2025simulating,wang2025decoding,gao2023s3}.
These comparisons, in turn, yield insights into the differences between opinion dynamics models and LLM-based simulations, for example, demonstrating that LLMs' opinion dynamics exhibit sharper changes than FJ models \cite{composta2025simulating} or that FJ and bounded confidence models are less capable of generating polarization compared to LLM-based simulations \cite{wang2025decoding}.
However, these works do not characterize which components of opinion dynamics models we require to explain LLMs' belief evolution.

\paragraph{Analytical models for opinion dynamics of networked LLMs}

A handful of recent works provide initial insights into how well classical opinion dynamics models fit the dynamics of networked LLMs \cite{yazici2026opinion,he2026opinion,abedini2026don,ju2024sense,khlytchievaanalyzing}.
In particular, \cite{yazici2026opinion} demonstrates that interacting LLM agents fail to converge to the consensus predicted by the DeGroot model.
Meanwhile, some prior works find that DeGroot or FJ dynamics can model the opinion dynamics of small groups of LLMs \cite{he2026opinion,abedini2026don}. 
Finally, \cite{ju2024sense} shows that the single-step dynamics of networked LLMs are well-modeled by Hegselmann-Krause dynamics.
While these studies point to correspondences between LLMs' dynamics and opinion dynamics, they fail to precisely identify which of the many possible components of opinion dynamics models accurately capture the opinions of networked LLMs. 
While more comprehensive explorations exist for two-agent interactions
\cite{brockers2025disentangling}, these analyses do not extend to networked interactions of many agents.

A few recent works propose opinion dynamics models for systems including LLMs, for example, modeling LLM-generated content's interaction with human opinions \cite{marchi2025boiling}, or opinion dynamics among humans who depend on LLMs \cite{li2026modeling}. 
Rather than proposing new opinion dynamics models, we explore how well existing models reflect networked LLMs' dynamics.

\section{Problem Statement}

We consider $n$ LLM agents that communicate over a directed graph $G = (V, E)$. Each vertex $i \in V$ corresponds to an LLM agent, and an edge $(i, j) \in E$ indicates that agent $j$'s messages are visible to agent $i$. 
We write $A \in \mathbb{R}_+^{n \times n}$ for the adjacency matrix of $G$, where the component $A_{ij}$ encodes the expected rate at which agent $i$ sees agent $j$'s posts. We use $\bar{A}$ to denote the row-normalized version of $A$. 
Each agent $i$ holds an opinion $z_i(t) \in [-1, 1]$, which evolves over discrete time $t \in \{0, 1, \ldots, H\}$ as the agents exchange messages, with $z(t) = [z_1(t)~\ldots~z_n(t)]^\top$.

We aim to characterize how the opinion vector $z(t)$ evolves in terms of the intrinsic behavioral traits of the constituent LLM agents, for example, their stubbornness, innate biases, or tendency toward homophilic interaction. To make this characterization tractable, we posit that $z(t)$ obeys a parametric model of the form
\begin{equation}\label{eq:parameterization}
    z(t+1) = f_\theta(z(t), z(0), A),
\end{equation}
where $\theta$ is a low-dimensional parameter vector encoding the agents' behavioral traits. The class~\eqref{eq:parameterization} subsumes classical opinion dynamics models such as DeGroot, Friedkin--Johnsen, and homophily-based dynamics, as well as extensions incorporating innate bias, which we instantiate in Section~\ref{sec:methods}.

Within this class, we further distinguish \emph{graph-specific} models, in which the parameters $\theta$ are fit independently for each graph, from \emph{graph-agnostic} models, in which $\theta$ encodes a small set of scalar coefficients that apply across topologies. The graph-agnostic form encodes the assumption that agent $j$'s contribution to agent $i$'s opinion update is a function of the rate at which agent $j$ messages agent $i$.
While restrictive, this assumption yields a behavioral profile $\theta$ that is portable across networks.

The recovered parameters $\theta$ act as a behavioral profile of the underlying LLM. By comparing these profiles across LLMs, topics, and graph structures, we aim to isolate which mechanisms are intrinsic to a model, which are topic-dependent, and which depend on the network in which the model is embedded, thus answering RQ1, RQ2, and RQ3.

\section{Methods}\label{sec:methods}

\subsection{Modeling Opinion Dynamics}
\label{sec:methods_models}

We now characterize the parametric models \eqref{eq:parameterization} that we consider in this study, each corresponding to a different hypothesis about the intrinsic characteristics of the LLMs that drive the evolution of their opinions. In all of the subsequent sections we make the notational assumption that multipliers of all additive terms, denoted as $\lambda_\star$, sum to one. This constraint is common across opinion dynamics models, and is needed to ensure stability and non-trivial opinion behavior.  Table \ref{tab:model_summary} summarizes the mechanisms represented by each model and the associated notation. 

\begin{table}[t]
\centering
\small
\setlength{\tabcolsep}{3pt}
\begin{tabular}{lcccl}
\toprule
Model & Stubb. & Homo & Bias & Parameters \\
\midrule
DG                  & -- & -- & -- & $W$ \\
FJ                  & \checkmark & -- & -- & $W,\lambda_0$ \\
DG + Homo           & -- & \checkmark & -- & $W,\gamma$ \\
FJ + Homo           & \checkmark & \checkmark & -- & $W,\lambda_0,\gamma$ \\
FJ + Bias           & \checkmark & -- & \checkmark & $W,\lambda_0,\lambda_b,b$ \\
FJ + Homo + Bias    & \checkmark & \checkmark & \checkmark &
$W,\lambda_0,\lambda_b,b,\gamma$ \\
Bias only*    & -- & -- & \checkmark &
$\lambda_b,b$ \\
Bias + Init only*    & \checkmark &  --  & \checkmark &
$\lambda_0,\lambda_b,b$ \\
\bottomrule
\end{tabular}
\caption{ Summary of the opinion dynamics models considered in this work. We abbreviate the models as follows:
    DG (DeGroot),
    FJ (Friedkin--Johnsen),
    Homo (Homophily weighting),
    Bias (Opinion bias term). 
    Asterisked models omit social interaction effects, and the ``Init'' abbreviation here refers to an initial opinion bias outside of an FJ model.
    }
\label{tab:model_summary}
\end{table}

\paragraph{DeGroot Model}

The simplest realization of the parametric model \eqref{eq:parameterization} is the DeGroot model, in which each agent's opinion is a convex combination of its neighbors' opinions, i.e.,
\begin{equation}\label{eq:DG}
z(t+1)=Wz(t).
\end{equation}
Here, $W$ is a row-stochastic matrix that plays the role of the parameter $\theta$, and it is such that $W_{ij} = 0$ if $(i,j) \notin E$ and $j \neq i$. Under some fairly mild assumptions \cite{golub2010degroot}, the DeGroot model is known to converge to a consensus in the steady state, i.e., to a situation where $z_1(t)=z_2(t)=\ldots=z_n(t)$ as $t\rightarrow\infty$.

\paragraph{Friedkin-Johnsen Model} In many situations, it is difficult for agents to reach a consensus as agents are stubborn and unwilling to abandon their initial opinions. The Friedkin-Johnsen (FJ) model captures this phenomenon by extending the DeGroot model \eqref{eq:DG} with a stubbornness term, leading to
\begin{equation}\label{eq:FJ}
    z(t+1) = \lambda_{\mathrm{avg}}W z(t) + \lambda_0 z(0).
\end{equation}
Here, $\lambda_{\mathrm{avg}}\in[0,~1]$ is the social averaging weight, $\lambda_{0}\in[0,~1]$ is the \textit{initial opinion weight} that specifies the resistance of each agent to deviate from their initial opinion, and $W$ is as in the DeGroot model. As previously stated, the multiplicative weights $\lambda_\star$ sum to $1$ by construction, i.e., $\lambda_{\mathrm{avg}} +\lambda_0 = 1$. The tunable parameter set $\theta$ comprises $\lambda_0$ and $W$. Note that due to the constraint on $\lambda_\star$, the averaging weight $\lambda_{\mathrm{avg}}$ is not a free parameter. Unless $\lambda_0=0$, the FJ model usually does not lead to a consensus.

\paragraph{Uniform Bias} 
LLM agents, in particular, often have a tendency to favor specific points of view on certain topics, owing to biases introduced during training \cite{pmlr-v202-santurkar23a}. These biases are independent of the initial opinions of the LLM agents and are thus not captured by the FJ model. To capture them, one can add a uniform bias term $b\in[-1,~1]$ to \eqref{eq:FJ}, leading to the opinion dynamics
\begin{equation}
    z(t+1) = \lambda_{\mathrm{avg}}W z(t) + \lambda_{0} z(0) + \lambda_b b,
\end{equation}
where $\lambda_b\in[0,~1]$ is the \textit{bias weight} indicating how biased agents are toward the opinion $b$, and as before, we have the constraint $\lambda_{\mathrm{avg}}+\lambda_0+\lambda_b=1$. In this case, $\theta$ comprises the parameters $\lambda_{0}, ~\lambda_b,~W$.

\paragraph{Homophily} While the addition of biases in the DeGroot model \eqref{eq:DG} captures a broader range of opinion dynamics, it misses a fundamental aspect of how social agents interact in that they weigh the opinions of like-minded agents more heavily. Homophily captures this kind of weighting by substituting the constant $W$ matrix with a distance-based one~\cite{newman2018networks}, namely
\begin{multline}\label{eq:homophily_term}
    H_{ij}(W, z, \gamma)=  \\ \frac{W_{ij} \exp(-\gamma |z_i - z_j| ) }{\sum_{j} W_{ij} \exp(-\gamma |z_i - z_j| )},
\end{multline}
where $\gamma$ controls the homophily strength. Higher $\gamma$ values force agents to be more selectively influenced by other agents of similar opinions. This new matrix leads to the opinion dynamics
\begin{multline}\label{eq:homophily}
    z(t+1) = \lambda_{\mathrm{avg}} H(W,z(t),\gamma)z(t)\\+ \lambda_0 z(0) + \lambda_b b,
\end{multline}
where $\lambda_\star$ sums to one as always, and the learnable parameters are $\gamma, \lambda_\star,b$, and $W$. We note that, within this work, homophily refers specifically to opinion-dependent influence weighting under comparable exposure, rather than similarity-based interaction partner selection.

\paragraph{Individual dynamics}
We can also omit social interaction effects entirely from the parametric model \eqref{eq:parameterization} by restricting $W$ to be the identity.
In this case, each agent's opinion follows dynamics of the form
\begin{equation}
    z(t+1) = \lambda_{\mathrm{avg}}I z(t) + \lambda_{\mathrm{bias}} b + \lambda_0 z(0).
\end{equation}
We term this decoupled model a bias-init-only model, and we use the term bias-only to refer to the case where $\lambda_0 = 0$.

\paragraph{Adjacency-Based Weighting} Finally, we discuss the role of the weighting matrix $W$ across all the models presented above. When the set of parameters $\theta$ in \eqref{eq:parameterization} includes the weighting matrix $W$, the fitted opinion dynamics model entangles the LLM agents' behavioral characteristics with the specific graph structure considered. To study whether the behavioral characteristics of LLMs transfer across different graph topologies, one can consider interaction matrices of the form
\begin{equation*}
\mu_{\mathrm{self}} I + \mu_{\mathrm{neigh}} \bar{A},
\end{equation*}
where $\bar{A}$ is the row-normalized adjacency matrix of the graph $G$, and $\mu_{\mathrm{self}}+\mu_{\mathrm{neigh}}=1$, $\mu_{\mathrm{self}},\mu_{\mathrm{neigh}}\in[0,1]$ are tunable parameters. 
In the case of homophily, the interaction term is the matrix
\begin{equation*}
    \mu_{\mathrm{self}} I + \mu_{\mathrm{neigh}} H(\bar{A}, z, \gamma). 
\end{equation*}
When inserted into the previously discussed opinion dynamics models, the parameters $\mu_{\mathrm{self}}$ and $\mu_{\mathrm{neigh}}$ get multiplied with $\lambda_{\mathrm{avg}}$ to denote weights corresponding to self update, and neighbor averaging, which we correspondingly denote as 
\begin{equation}
    \lambda_{\mathrm{self}} = \lambda_{\mathrm{avg}}\mu_{\mathrm{self}}\; \textrm{, and } \; \lambda_{\mathrm{neigh}} = \lambda_{\mathrm{avg}}\mu_{\mathrm{neigh}}.
\end{equation}
For the  DeGroot model, with no $\lambda_\star$ parameters, we set $\lambda_{\mathrm{self}} = \mu_{\mathrm{self}}$, and $\lambda_{\mathrm{neigh}} = \mu_{\mathrm{neigh}}$.
For bias-only and bias-init-only models, we enforce that $\mu_{\mathrm{self}} = 1$ and $\mu_{\mathrm{neigh}} = 0$, so that $\lambda_{\mathrm{self}} = \lambda_{\mathrm{avg}}$.

\subsection{Simulating LLMs' Opinion dynamics}

\subsubsection{Discussion protocols}

To generate opinion dynamics data for groups of LLMs, we simulate conversations amongst LLM-agents on synthetically generated graphs.

Agents post in sequence with random inter-post times and ordering.
For each post, we select an agent uniformly randomly to generate the next post, with that post being added to a pool of posts for that agent.
We then assign random times to each post according to a Poisson process with rate $\rho$.

To generate the posts, we prompt the agents to construct arguments in response to other agents' posts.
Fix an agent $i \in V$ and an integer parameter $k$.
We first select up to $k$ posts uniformly at random from the prior posts of the agents that influence $i$.
We then prompt the agent to justify their position based on 1) their initial opinion, and 2) the $k$ sampled posts.
Finally, we prompt the agent to generate a post based on their internal justification.

We encode the set of agents influencing agent $i$ in the interaction graph $G$, where edge $(i,j)$ indicates that agent $j$ influences agent $i$.
As we sample the $k$ posts that agent $i$ sees uniformly from agent $i$'s neighbors, and $A_{ij}$ encodes the rate at which $i$ sees $j$'s posts, the entry $\bar{A}_{ij}$ of the row-normalized adjacency matrix is equal to $\nicefrac{1}{N(i)}$ for each neighbor $j$ of $i$, where $N(i)$ is the number of agents influencing agent $i$.
We next detail the generation of the interaction graph $G$.

\subsubsection{Network generation}
\label{sec:network_generation}

We evaluate the opinion dynamics of LLMs on a mixed dataset composed of three widely studied random graph models that together cover the topological phenomena commonly attributed to online social networks, such as small-world effects and community structures.

\paragraph{Erd\H{o}s--R\'enyi models \cite{bollobas2001randomgraphs}}
The Erd\H{o}s--R\'enyi model is a widely used baseline for random graph analysis. Each pair of nodes is connected independently with a uniform probability $p$,  producing a maximally unstructured network with homogeneous degree distribution and no latent community organization.
\paragraph{Chung-Lu models \cite{Chung2002chunglu}}
The Chung-Lu model builds upon the Erd\H{o}s--R\'enyi model by assigning each node an expected degree, with edge probabilities set proportional to the product of endpoint degree weights. 
When we randomly sample these node weights, Chung-Lu models reproduce degree heterogeneity, incorporating hubs vs. low-degree periphery nodes.

\paragraph{Stochastic block models  \cite{holland1983sbm}}
Stochastic block models (SBMs) partition nodes into communities and define edge probabilities as a function of community membership, with within-block connectivity being typically denser than cross-block connectivity. Thus, SBMs are the standard generative model for simulating networks with community structure, capturing the clustering patterns characteristic of real social networks such as echo chambers or professional affiliation groups.

For each run used to answer RQ1/RQ2, we sample a graph from a randomly selected type. 
For each LLM-topic pair used to answer RQ3, we generate a single graph on which we simulate all results.

\subsubsection{Opinion measurement}
We convert the simulation output, a sequence of natural-language posts indexed in continuous time, into discrete-time-indexed numerical opinions.

We compute opinions by assessing similarity in representation space between agents' posts and reference prompts.
Fix a post $t$, and two texts $a_{\text{pro}}$ and $a_{\text{anti}}$ that encode for and against opinions on a given topic.
We first pass $t$, $a_{\text{pro}}$, and $a_{\text{anti}}$ through a stance-aware sentence transformer \cite{ghafouri-etal-2024-love}  to get embedding vectors $\mathbf{t}$,
$\mathbf{a}_{\text{pro}}$, and $\mathbf{a}_{\text{anti}}$, respectively. 
We then compute a post's opinion as
\begin{equation}
    z = \text{sim}(\mathbf{t}, \mathbf{a}_{\mathrm{pro}}) - \text{sim}(\mathbf{t}, \mathbf{a}_{\mathrm{anti}}),
\end{equation}
where $\text{sim}(\mathbf{x}, \mathbf{y}) = \nicefrac{\mathbf{x} \cdot \mathbf{y}}{\|\mathbf{x}\| \|\mathbf{y}\|}$.
In Appendix~\ref{sec:stance_model_validation}, we examine the accuracy of this measurement method and its sensitivity to the texts $a_{\text{pro}}$ and $a_{\text{anti}}$.

We sample the agent opinions at discrete timestamps to extract the ground-truth discrete-time evolution of agent opinions. 
To account for startup latency during network initialization, we treat the timestamp of the first post as the start time of the simulation. 
We then assign subsequent posts a simulated timestamp drawn from the Poisson-process wait times.
We partition time into discrete bins corresponding to timesteps and compute an agent's opinion at each timestep based on the posts within the corresponding bin. If an agent has no posts during a given timestep's bin, we carry forward their opinion from the previous timestep. 
If an agent has multiple posts during a given timestep, we use their last post.
When an agent has yet to publish a message, we score the agent's initialization prompt, defining its assigned opinion on the topic.

\subsection{Opinion initialization}

We initialize each agent with a discrete opinion from the set $\{-1, -0.5, 0, 0.5, 1\}$  ranging from complete disagreement to complete agreement on the topic, where $0$ denotes a neutral stance. We map each value to a natural-language descriptor, e.g.,``opposed", ``skeptical", ``agree", and we embed it in the agent's initialization prompt, which states the agent's stance on the topic.

To initialize the network, we randomly assign agents to the five opinion categories according to predefined probability distributions. Each distribution corresponds to a set of initial conditions, for example, uniform, polarized, or skewed configurations.
For RQ1 and RQ2, we allocate a training set for model fitting, with the remaining runs being allocated as a test set that we use in answering RQ1, 
and we follow a similar procedure for RQ3, in each case generating train and test data using various initial opinion distributions.

\subsection{Model Fitting}

\begin{figure*}[h]
    \centering
    \includegraphics[width=\linewidth]{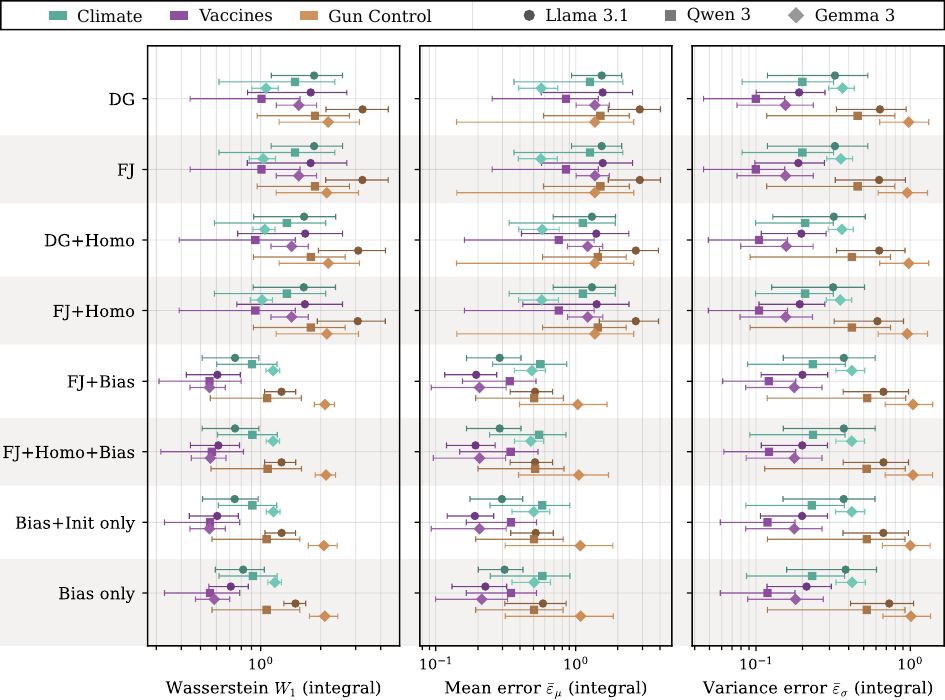}
    \caption{Out-of-sample prediction error across opinion dynamics models, LLMs, and discussion topics, reported on three metrics. Each point shows the mean error across held-out runs.  Whiskers show one standard deviation. We abbreviate the models as follows:
    DG (DeGroot),
    FJ (Friedkin--Johnsen),
    Homo (Homophily weighting),
    Bias (Opinion bias term).
    We then use ``Bias+Init only'' and ``Bias only'' to denote the models discussed in Section~\ref{sec:methods_models}  that omit social interaction effects.
    Models that include a uniform bias term consistently achieve the lower errors across mean and Wasserstein-1 distance metrics.}
    \label{fig:rq1}
\end{figure*}

We fit the opinion dynamics model by minimizing the following least squares objective
\begin{equation}
    \label{eq:ls}
    \sum_{t,r}|| z^{(r)}(t+1) - f_\theta(z^{(r)}(t), z^{(r)}(0), \bar{A})||_2^2,
\end{equation}
where for each model, the opinion dynamics model $f_\theta$ admits the dynamics in the corresponding model definitions in Section~\ref{sec:methods_models}.

\subsubsection{Fully parametrized models}
In fully parametrized models, the influence matrix $W$ is an $n \times n$ row-stochastic parameter fit independently for each graph instance with entries constrained to respect the graph $G$. For models without a homophily term, i.e., DeGroot, Friedkin-Johnsen, and Friedkin-Johnsen with bias, the objective is convex in the model parameters. For fixed $\lambda_0$ and $\lambda_b$, optimizing over $W$ is a constrained quadratic program (QP), which is solvable using a standard convex solver. For the Friedkin-Johnsen-with-bias model, we apply the reparameterization $\tilde{W} = \lambda_{\mathrm{avg}}W$ and $\tilde{b} =\lambda_{b} b$, under which the joint optimization over $\tilde{W}$, $\lambda_0$, and $\tilde{b}$ is a convex QP solvable to global optimality given sufficient data (see Appendix \ref{apx:cvx_reformulation}).

For the models containing homophily, the influence matrix $H(W, z, \gamma)$ depends nonlinearly on the opinion vector $z(t)$ through the exponential distance kernel, making the objective non-convex in $(W, \gamma)$ jointly. We therefore optimize these models using a nonlinear SLSQP solver in SciPy with multiple random restarts \cite{2020SciPy-NMeth}.

\subsubsection{Adjacency-based models}
In adjacency-based models, we constrain the interaction matrix to be of the form $\mu_{\mathrm{self}} I + \mu_{\mathrm{neigh}} \bar{A}$ or $\mu_{\mathrm{self}} I + \mu_{\mathrm{neigh}} H(\bar{A}, z, \gamma)$, {\color{blue} or $I$}, 
reducing the parameter space to a small set of scalar coefficients shared across graphs. For models without homophily, fitting these coefficients
reduces to a convex QP in the $\lambda$'s, which we solve.

For models that include homophily, the homophily strength parameter $\gamma$ enters the objective nonlinearly and cannot be recovered by convex methods. We instead perform a one-dimensional sweep over $\gamma$, where,
for each candidate value, we solve the resulting convex QP for the $\lambda_\star$'s and select
the $\gamma$ minimizing the objective~\eqref{eq:ls}.

Given their parameter efficiency and generalizability, we investigate RQ1 and RQ2 with adjacency-based models, after which we compare them against fully-parameterized variants in RQ3.

\section{Experimental Results}
\label{sec:experiments}

\begin{table*}[h]
  \centering
  \small
  \caption{%
    Fitted parameters of the \textit{FJ}+\textit{Homo}+\textit{Bias} model
    across LLMs and topics.
    $\lambda_{\mathrm{neigh}}$ and $\lambda_{\mathrm{self}}$ are the
    neighbor influence and self-weight terms;
    $\lambda_{0}$ is the initial opinion weight;
    $b$ and $\lambda_b$ are the opinion bias and its weight;
    $\gamma$ controls homophily strength.
     ($\dagger$) The parameter sweep for homophily parameter $\gamma$ for \textsc{Llama3.1} in Vaccines topic hit an upper limit of $75.00$.
  }
  \label{tab:fitted_params}
  \begin{tabularx}{\textwidth}{llSSSSSSW}
    \toprule
    LLM & Topic & {$\lambda_{\mathrm{neigh}}$} & $\lambda_{\mathrm{self}}$ & $\lambda_0$ & $\lambda_b$ & $b$ & $\gamma$ & $\lambda_{\mathrm{neigh}}/(\lambda_{b}b)$ \\
    \midrule
    \multirow{3}{*}{\textsc{Gemma3}} & Climate & $0.0812$ & $0.6061$ & $0.0051$ & $0.3075$ & $0.3452$ & $1.4134$ & $0.7650$ \\
     & Gun Control & $0.1647$ & $0.6621$ & $0.0017$ & $0.1715$ & $0.4019$ & $0.3444$ & $2.3890$ \\
     & Vaccines & $0.0161$ & $0.5985$ & $0.0147$ & $0.3707$ & $0.5679$ & $4.3954$ & $0.0767$ \\
    \midrule
    \multirow{3}{*}{\textsc{Llama3.1}} & Climate & $0.0182$ & $0.6530$ & $0.0203$ & $0.3086$ & $0.6751$ & $0$ & $0.0873$ \\
     & Gun Control & $0.0041$ & $0.6686$ & $0.0383$ & $0.2891$ & $0.8513$ & $0$ & $0.0165$ \\
     & Vaccines & $0.0363$ & $0.6671$ & $0.0353$ & $0.2612$ & $0.5973$ & {$75.00^{\dagger}$} & $0.2327$ \\
    \midrule
    \multirow{3}{*}{\textsc{Qwen3}} & Climate & $0.0996$ & $0.6969$ & $0.0040$ & $0.1996$ & $0.5423$ & $1.1571$ & $0.9205$ \\
     & Gun Control & $0.0512$ & $0.6798$ & $0$ & $0.2690$ & $0.6363$ & $0.7092$ & $0.2990$ \\
     & Vaccines & $0.1724$ & $0.6065$ & $0.0015$ & $0.2195$ & $0.4866$ & $2.6527$ & $1.6141$ \\
    \bottomrule
  \end{tabularx}
\end{table*}

We evaluate our framework on three open-weight LLM families, \textsc{Llama3.1},
\textsc{Qwen3}, and \textsc{Gemma3}, across three contested discussion topics:
climate change, vaccines, and gun control. We choose these topics because they
provoke sustained disagreement and elicit clear stances. For each LLM and topic, we simulate
populations of $n=30$ agents exchanging posts over the random graphs we describe in
Section~\ref{sec:network_generation}, fit the opinion dynamics models of
Section~\ref{sec:methods} to the resulting trajectories, and measure how well
each model predicts opinions on held-out runs. We organize our findings around
the three research questions of Section~\ref{sec:intro}.

\subsection{RQ1: Which components drive LLM opinion dynamics?}
\label{sec:experiments_RQ1}
Figure~\ref{fig:rq1} reports out-of-sample error for all eight models across every LLM and topic. The bias-free models, namely DeGroot, Friedkin-Johnsen, and their homophily variants, cluster together at high error, while the four models that include a uniform bias term separate from this cluster and attain the lowest mean and Wasserstein-1 errors in nearly every setting, while leaving variance error relatively unchanged. Adding bias reduces integral mean error by up to $88\%$ and integral Wasserstein-1 distance by up to $66\%$, both attained on \textsc{Llama3.1}. 

Further comparisons reinforce the innate model bias as the primary model component responsible for the majority of error reduction. First, Friedkin-Johnsen tracks DeGroot almost exactly in every error metric, which shows that anchoring agents to their own initial opinions does not improve fit. Second, adding homophily weighting to any base model changes its error by at most a few percent, and once bias is present homophily produces no measurable change. The mechanism that matters is therefore neither stubbornness toward one's initial opinion nor preferential weighting of like-minded neighbors, but regression toward a single opinion shared across all agents. 
The recovered bias term should be interpreted as an effective modeling component rather than a directly interpretable behavioral trait. It may capture multiple coupled effects, such as model priors, prompting effects, alignment behavior, topic-dependent language patterns, and properties of the stance measurement procedure, which we do not disentangle in these experiments.

The size of the bias effect depends on the LLM. \textsc{Llama3.1} shows the largest and most uniform gain, with integral mean error falling between $81\%$ and $88\%$ across all three topics. \textsc{Qwen3} gains the least, with reductions between $56\%$ and $66\%$. \textsc{Gemma3} sits between the two and is the only model for which bias fails to help in one setting, where climate error rises slightly because the bias-free fit is already accurate.

The bias effect also depends on the topic. 
The vaccines topic tends to benefit the most from the bias term, with integral mean error reductions of $60\%$ to $88\%$ and Wasserstein-1 reductions of $45\%$ to $66\%$, exceeding the corresponding reductions for climate and gun control in most LLMs. 
A plausible explanation is that alignment training drives LLMs toward a consistent stance on health misinformation regardless of assigned initial opinion, an effect that the bias term captures directly. 

We additionally remark that the social-interaction-free baseline performs remarkably well across many topics and families, attaining the best integral mean or Wasserstein-1 errors on multiple topic-family pairs.  
However, the dominance of this naive baseline is not uniform, and we emphasize that the optimal parameter values include non-zero social effects, reflecting that accounting for neighbor interactions provides predictive value beyond simply regressing toward a bias determined by the topic, model, and initialization.

We also note that the value of bias as a modeling component holds in larger networks with $n=60$ agents, which we report in Appendix~\ref{sec:size_variation}.

\textbf{Takeaway.} Across the evaluated LLMs and topics, the bias term provides the largest improvement in predictive accuracy among the mechanisms considered, while stubbornness and homophily provide comparatively little additional improvement.

\subsection{RQ2: Parameter Comparison}

Having established that \textit{FJ} dynamics with added \textit{Opinion Bias} capture LLM opinion dynamics, a natural question is what quantitative weight each opinion dynamics component has on the opinion evolution. This section aims to quantify these weights by comparing fitted model parameters directly, and comparing how the parametrizations differ across LLM/topic combinations.

Table~\ref{tab:fitted_params} shows the parameters of the full \textit{FJ}+\textit{Homo}+\textit{Bias} model fit separately to each LLM and topic. The self-weight $\lambda_{\mathrm{self}}$ is the largest parameter in every case, between $0.60$ and $0.70$, so each agent's next opinion is governed primarily by its current one. The scale of the self weights, however, is largely irrelevant to the long-term behavior of the LLMs, as $\lambda_{\mathrm{self}}$ only determines the convergence speed in comparison to the sampling intervals. Thus, the scale of the self weights is purely an artifact of the chosen time scale. 

The initial-opinion weight $\lambda_0$ never exceeds $0.04$, confirming across every LLM and topic that anchoring to starting opinion plays very little role in driving opinion evolution. The bias weight $\lambda_b$ is the second-largest term in most settings, between $0.17$ and $0.37$, so regression toward an innate opinion is a consistent feature of LLM dynamics rather than an artifact of a single model or topic.

The ratio $\lambda_{\mathrm{neigh}}/(\lambda_b b)$ between neighbor opinion effects and self-bias effects provides an overall summary of how much each LLM family preferentially allocates its opinion updates to social interaction effects. This ratio shows substantial variance between LLM/topic combinations. \textsc{Llama3.1} holds it near zero on every topic, from $0.02$ on gun control to $0.23$ on vaccines, so its bias dominates its opinions and its neighbors barely register. \textsc{Qwen3} and \textsc{Gemma3} tend to show higher social influence than \textsc{Llama3.1} in five out of six test cases, with the $\lambda_{\mathrm{neigh}}/(\lambda_b b)$ ratio running between $0.30$ and $1.61$ for \textsc{Qwen3}, and \textsc{Gemma3}'s ratio fluctuating with the topic, from $0.08$ on vaccines to $2.39$ on gun control, the largest value in the table. How social an LLM acts thus depends on the topic as much as on the model, and only \textsc{Llama3.1} stays at one extreme across all three, ignoring social interaction effects almost completely in its long-term opinion evolution. This behavior may be attributable to \textsc{Llama3.1} being an older model, as recent improvements in post-training may teach models to weigh the conversation in front of them over a stance fixed during training, although three models cannot establish the cause.

Analyzing the sensitivity of error to the $\lambda_\star$ and $b$ terms through a second-order derivative analysis shows that the measurement error is most sensitive to changes in $\lambda_0$, with a measured second derivative of $\frac{\partial^2 \mathrm{MSE}}{\partial \lambda_0^2} = 0.27$ around the optimal fit (see Appendix~\ref{app:hessian}) for \textsc{Gemma3} on the climate topic. 
Furthermore, across LLMs and topics, the second derivative for a majority of parameters, excluding $\gamma$, remains above $0.1$.
The high curvature suggests that estimated parameters are relatively robust to errors in opinion estimation and loss computation.

In contrast to the weight values and bias value, homophily barely affects the fit under the randomized interaction protocol used in this work. Figure~\ref{fig:homophily_sweep} sweeps $\gamma$ over its full range for a subset of LLM-topic pairs, and the error varies nonlinearly by at most a percent. Accordingly, the strength of homophily in the majority of experiments is low.

\textbf{Takeaway.} The fitted balance between neighbor influence and bias varies substantially across LLM--topic pairs, while initial-opinion anchoring remains small and homophily is weakly identified under the interaction protocol considered.

\begin{figure}
    \centering
    \includegraphics[width=1\linewidth]{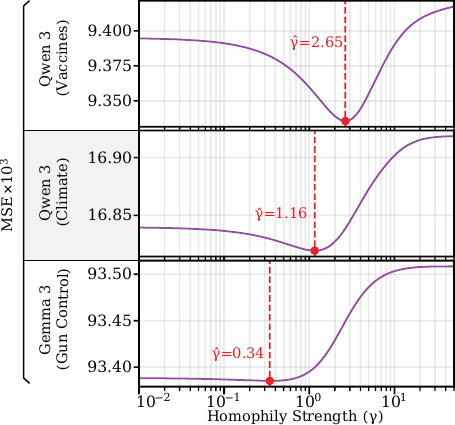}
    \caption{Fit error as a function of homophily strength $\gamma$ for three representative LLM--topic pairs. The dashed red line marks the optimal $\gamma$ minimizing the error.}
    \label{fig:homophily_sweep}
\end{figure}

\subsection{RQ3: Full Parametrization vs. Adjacency weighting}

On a single fixed graph, we compare a fully parameterized per-edge weight
matrix $W$ against the adjacency-based form, which derives influence from the communication structure using only a handful of scalar coefficients. 
Table~\ref{tab:w_parameterization} compares the two parameterizations on
training error and on two held-out metrics. While the fully
parameterized $W$ lowers the training error, 
test metrics show either minor or no improvements, demonstrating that the full parameterization of $W$ fails to provide a
significant modeling benefit. 
The homophily variant reproduces the
FJ+Bias numbers in the \textsc{Llama3.1}-Climate case because the fitted homophily strength collapses to zero
in the local optima we find. By absorbing the graph into a known input
and fitting only scalar coefficients, the adjacency-based form recovers
a behavioral profile of the LLM that is not tied to the specific network
it was fitted on. 

\textbf{Takeaway.} The recovered opinion dynamics are largely preserved across network structures.

\begin{table}[h]
  \centering
  \small
  \caption{%
    Effect of fully parameterizing $W$ against an adjacency-based form.
    MSE$_{\mathrm{tr}}$ is training error; $W_1$ and
    $\bar\varepsilon_\mu$ are held-out test errors. We show the best per column for
    each family in \textbf{bold}.
  }
  \label{tab:w_parameterization}
\renewcommand{\arraystretch}{0.9}
  \begin{tabularx}{\columnwidth}{@{}lMTT@{}}
    \toprule
     & MSE$_{\mathrm{tr}}$ $\times10^{3}$& $W_1$ & $\bar\varepsilon_\mu$ \\
    \midrule
    \multicolumn{4}{@{}l}{\textsc{Gemma3} Vaccines}\\
    \multicolumn{4}{@{}l}{\quad \textsc{FJ} + \textsc{Bias}} \\
    \quad \quad Adjacency & 14.03 & 0.616 & 0.207 \\
    \quad \quad Full $W$ & 13.53 & $\mathbf{0.613}$ & $\mathbf{0.192}$ \\
    \addlinespace[2pt]
    \multicolumn{4}{@{}l}{\quad $+\,$\textsc{Homo}} \\
    \quad \quad Adjacency & 14.01 & 0.627 & $\mathbf{0.209}$ \\
    \quad \quad Full $W$ & $\mathbf{13.39}$ & $\mathbf{0.624}$ & 0.213 \\
    \midrule
    \multicolumn{4}{@{}l}{\textsc{Llama3.1} Climate}\\
    \multicolumn{4}{@{}l}{\quad \textsc{FJ} + \textsc{Bias}} \\
    \quad \quad Adjacency & 23.85 & $\mathbf{0.827}$ & $\mathbf{0.249}$ \\
    \quad \quad Full $W$ & 23.23 & 0.862 & 0.310 \\
    \addlinespace[2pt]
    \multicolumn{4}{@{}l}{\quad $+\,$\textsc{Homo}} \\
    \quad \quad Adjacency & 23.85 & $\mathbf{0.827}$ & $\mathbf{0.249}$ \\
    \quad \quad Full $W$ & $\mathbf{23.23}$ & 0.862 & 0.310 \\
    \bottomrule
  \end{tabularx}
\end{table}

\section{Conclusion}
 We study which components of classical opinion dynamics models best explain how opinions evolve in networked LLMs, finding that an effective bias term provides the largest improvement in predictive accuracy among the mechanisms considered. Stubbornness and opinion-dependent influence homophily contribute little, while the combination of social averaging and the effective bias term accounts for most of the gain in fit across three model families, three contested topics, and several network structures. Interpreting the fitted parameters as descriptive profiles of opinion dynamics separates the models by how much they weigh their neighbors against the effective bias term, with \textsc{Llama3.1} almost purely dispositional and \textsc{Qwen3} the most social.  These results provide an exploratory characterization of LLM opinion dynamics under the simulated interaction settings considered here, providing a foundation for further study of networked LLMs. 

Our findings suggest two practical considerations for researchers studying opinion evolution in LLM networks. First, models of LLM opinion dynamics should explicitly account for model- and topic-dependent biases, as excluding this component substantially reduces predictive accuracy in our experiments. Second, LLM-based simulations should not be assumed to naturally reproduce mechanisms underlying human opinion dynamics. Instead, mechanisms observed in human populations should be empirically validated when using LLM agents as models of social systems.

\section*{Limitations}
Our findings should be interpreted in the context of the following limitations.

\paragraph{Limited number of evaluation runs.} For each LLM-topic pair in RQ1/RQ2, we evaluate fits on a held-out test set of only eight initial-condition configurations. This number of runs proves sufficient to surface the systematic effects we report, such as the dominant role of innate bias in LLM agents.

\paragraph{Relatively small population size.} All simulations use $n = 30$ or $n = 60$ agents. Though small relative to real online discussions, this scale is enough to exhibit nontrivial network effects across diverse topologies. In particular, the homogeneity of the agents in our setup means that opinion dynamics at larger scales should generally exhibit similar patterns. 

\paragraph{Restricted set of LLMs.} We study three open-weight LLM families, namely Qwen, Llama, and Gemma, chosen for the feasibility of running large numbers of parallel agents on modest hardware. Although more models exist, the consistency of our findings across these three families suggests that the identified mechanisms reflect properties common to instruction-tuned LLMs rather than idiosyncrasies of a single model.

\paragraph{Synthetic interaction networks.} Our experiments use synthetically generated communication graphs to isolate the effects of different opinion-dynamics mechanisms. While this controlled setting enables systematic comparison between models, real LLM-mediated interaction networks may exhibit additional properties, such as evolving connectivity and heterogeneous user behavior. Therefore, our conclusions about the relative importance of bias, stubbornness, and homophily should be interpreted within the simulation protocol considered here. The proposed fitting framework, however, is not limited to synthetic networks and can be applied to empirical interaction trajectories when such data become available.

\paragraph{Prompting protocol.} Our results are conditioned on the prompting protocol used to initialize and update the LLM agents. Because agents are repeatedly reminded of their assigned stance, the recovered contributions of bias and social influence may depend on this design choice. Evaluating the sensitivity of the fitted parameters to alternative prompting strategies remains an important direction for future work.

\section*{Ethical Considerations}
This work characterizes how the opinions of interacting LLMs evolve over time. We study these mechanisms to make such dynamics legible and governable rather than to optimize persuasion, and we report no methods for steering live human discourse. We nonetheless acknowledge that understanding a mechanism is a step toward exploiting it. A malicious operator could use opinion dynamics methods to predict opinion evolution in social networks of mixed human and LLM populations, and steer that population toward a target view.  However, our methods are also equally beneficial for predicting runaway consensus, which enables platforms to detect and dampen it before it takes hold.

Our experiments use contested topics, namely climate change, vaccines,
and gun control, and they recover model-specific biases on each. These estimates describe the behavior of particular model versions in our simulation setup. They are not endorsements of any position, fixed properties of the models, or claims about the topics themselves. While we kept the methodology and prompting neutral and unbiased throughout our experiments, model behavior changes across versions and prompting choices, so the values we report may not transfer to other deployments.

Finally, our agents are LLMs simulating opinion holders, not humans. We make no claim that the recovered dynamics are representative of human opinion formation, and we caution against using LLM populations as drop-in surrogates for human ones without independent validation.

\section*{Acknowledgments}
This work was supported in part by the United States Army Research Office under grant W911NF-23-1-0317 and the Office of Naval Research under grant N00014-24-1-2432.

\bibliography{llm_network_fitting}

\newpage
\appendix

\clearpage
\section{Appendix}
\label{sec:appendix}

\subsection{Computational details}

\subsubsection{Large Language Model Families}

We simulate opinion dynamics using the following classes of models.
\begin{itemize}
    \item Llama-3.1-8B \cite{grattafiori2024llama}, which we implement through the Ollama interface\footnote{Accessible from: \url{https://ollama.com/library/llama3.1:8b}}.
    \item Qwen3-VL-8B-Instruct \cite{yang2025qwen3technicalreport}, which we implement through the Hugging Face interface \footnote{Accessible from: \url{https://huggingface.co/Qwen/Qwen3-VL-8B-Instruct}}.
    \item Gemma-3-4b-it \cite{gemmateam2025gemma3technicalreport}, which we implement through the Ollama interface\footnote{Accessible from: \url{https://ollama.com/library/gemma3:4b}}.
\end{itemize}

\subsubsection{Stance classification models}

We use the StanceAware-SBERT model \cite{ghafouri-etal-2024-love} for stance classification via the SentenceTransformers interface 
\footnote{\url{https://sbert.net/}}.

\subsubsection{Computational details}

We generated simulations for research questions $1$ and $2$ using the machine detailed in Table~\ref{tab:compute_sim}.
LLMs were run on the GPU, and each simulation run took approximately $30$-$40$ minutes to generate. At $40$ minutes per run, with 
\begin{equation}
    \underbrace{3 \times 3 \times 32}_{\text{3 topics, 3 LLMs for RQ2}} = 288
\end{equation}
runs, we estimate the total time used on this machine at approximately $192$ hours.
\begin{table}[H]
\centering
\caption{Compute resources used for simulations for RQ1 and RQ2.}
\label{tab:compute_sim}
\begin{tabular}{ll}
\hline
\textbf{Resource} & \textbf{Specification} \\
\hline
CPU          &  Intel i9-14900K \\
Cores        &  24 \\
System Memory  &  64 GB \\
GPU          &  Nvidia RTX 4080 \\
GPU Memory         &  16 GB \\
\hline
\end{tabular}
\end{table}

We used the machine detailed in Table~\ref{tab:swarm_compute_sim} to generate simulations for research question $3$ and to generate the experiments with larger networks that we detail in Section~\ref{sec:size_variation}.
LLMs were run on the GPUs, and each simulation run took approximately $30$-$40$ minutes to generate. 
At $40$ minutes per run, with 
\begin{multline}
    \underbrace{42 \times 2}_{\text{2 LLM-topic pairs for RQ3}} \\+ \underbrace{32}_{\text{1 topic-LLM pair with 60 agents.}} = 116 
\end{multline}
runs, we estimate the total time used on this machine at approximately $78$ hours.
We remark that we omitted training runs from fitting when they stopped prematurely.
However, the estimate above includes these omitted runs.

\begin{table}[H]
\centering
\caption{Compute resources used for simulations for RQ3, as well as $60$-agent experiments.}
\label{tab:swarm_compute_sim}
\begin{tabular}{ll}
\hline
\textbf{Resource} & \textbf{Specification} \\
\hline
CPU          &  Intel Xeon Gold 6226R CPU\\
Cores        &  $32$ ($64$ threads)\\
System Memory  & $252$ GB \\
GPU          & $8 \times $ Nvidia RTX A5000  \\
GPU Memory   & $8 \times $ 24 GB \\
\hline
\end{tabular}
\end{table}

We fit the opinion dynamics models and estimated stances for each research question
using the machine that we detail in Table~\ref{tab:compute_fit}.
Computational time on this machine was minimal, requiring approximately one hour for all embeddings and two hours to carry out model fitting for all research questions.
\begin{table}[H]
\centering
\caption{Compute resources used for model fitting.}
\label{tab:compute_fit}
\begin{tabular}{ll}
\hline
\textbf{Resource} & \textbf{Specification} \\
\hline
CPU         &  Intel i7-13700H \\
Cores       & 20 \\
System Memory  &  32 GB \\
GPU          & Nvidia RTX 4060 \\
GPU Memory   & 8 GB \\
\hline
\end{tabular}
\end{table}

\subsection{Details on LLM network simulation}

\subsubsection{Interaction algorithm}

To generate opinion dynamics data for groups of LLMs, we simulate conversations in English amongst LLM-agents on synthetically generated graphs.
In particular, we follow Algorithm~\ref{alg:posting_protocol}.

 \begin{algorithm}
 \caption{Agent Posting Protocol}
 \label{alg:posting_protocol}
 \begin{algorithmic}[1]
     \Require Agent $i \in V$, initialization prompt $\pi_i$, feed $\mathcal{F}_i \subseteq \mathcal{N}^{\text{in}}(i)$
     \Ensure Post $\rho_i^{(t)}$ published to network at simulated time $t$
    \Loop
         \State \textbf{await} new message on Redis stream $\mathcal{S}$
         \If{$i \neq \textit{NextAgent}$}
             \State \textbf{continue}
         \EndIf
         \State Acquire mutex $\mu$ 
         \State $\textit{NextAgent} \leftarrow \texttt{OrderManager}.\textsc{Next}()$ 
         \State $\Delta t \sim \texttt{TimeManager}.\textsc{Sample}()$
         \State \textbf{wait} $\Delta t$ 
         \State $\mathcal{C}_i \leftarrow \textsc{FetchPosts}\!\left(\mathcal{F}_i,\; k{=}10\right)$ 
         \State $\rho_i^{(t)} \leftarrow \textsc{LLM}\!\left(\pi_i,\; \mathcal{C}_i\right)$ 
         \State Publish $\rho_i^{(t)}$ to $\mathcal{S}$
         \State Release mutex $\mu$
     \EndLoop
 \end{algorithmic}
 \end{algorithm}

\subsubsection{Prompts for post generation}
We initialize agents with a prompt template to give each agent an opinion towards a target topic while encouraging social-network style content generation. Each agent is assigned an initial opinion value from the set $\{-1, -0.5, 0, 0.5, 1\}$, corresponding to a level of disagreement or agreement on a topic. These values are mapped to the following prompt template. 

\begin{quote}
You are a social media user who posts about {topic}. There are other users on the network who have different perspectives on this topic.

The sentence $``{unique\_prompt}"$ reflects your stable perspective on this topic.
\end{quote}

The variable \textit{unique\_prompt} is instantiated a template corresponding to a stance in $\{-1, -0.5, 0, 0.5, 1\}$.
We depict example templates below.

\begin{quote}
    ``I am a social media user who is completely opposed to the statement, ${statement}"$"
    
    ``I am a social media user who is skeptical of the statement, ${statement}$"
    
    ``I am a social media user who is neutral about the statement, ${statement}$"
    
    ``I am a social media user who tends to agree with the statement, ${statement}$"
    
    ``I am a social media user who strongly believes in the statement, ${statement}$"
\end{quote}

At each posting step, agents generate posts according to the following procedure:

\begin{enumerate}
    \item Reaffirm the existing prompt template used for initial opinion assignment. 

    \item Sample up to 10 posts uniformly at random from neighboring agents that influence the current agent according to the graph structure.

    \item Append stylistic constraints to the prompt, including instructions to avoid emojis, excessive special characters, hashtags, and engagement-bait phrasing.

    \item Retrieve the agent’s rolling history of internal justifications generated during previous posting steps.

    \item Generate a short internal justification conditioned on the initialization opinion, sampled neighboring posts, and prior justification history. The agent is encouraged to produce new reasoning, examples, or claims while avoiding repetition of earlier arguments. The generated justification is appended to the rolling history, which is truncated to retain only recent entries.

    \item Generate the final social media post conditioned on the internal justification generated in the previous step.
\end{enumerate}

We next give examples of generated posts.

\subsubsection{Example posts}

\textsc{Gemma3} on vaccines and autism:
\begin{quote}
    "The numbers don’t lie. A recent German study of 92 children who died before age five, all vaccinated within 12 months, revealed an average aluminum concentration of 198 micrograms per gram – over 400\% above safe levels. This isn’t a fluke; it’s the consistent finding across multiple investigations examining brain tissue. The continued use of aluminum adjuvants in vaccines is undeniably creating a chronic inflammatory assault on developing brains, leading to devastating neurological outcomes. I’ve seen firsthand the impact – children robbed of their potential, their futures irrevocably altered by this silent poison. It’s time we acknowledge the truth before more young lives are lost."
\end{quote}

\textsc{Llama3.1} on climate change:
\begin{quote}
     "A recent study published in the Journal of Environmental Science found that atmospheric CO2 levels have correlated with El Niño events over the past several decades. This relationship suggests that natural climate variability may play a larger role in global temperature fluctuations than previously thought, complicating attempts to attribute climate change solely to human activities."
\end{quote}

\textsc{Qwen3} on the gun-control topic:
\begin{quote}
    "The real tragedy isn’t the gun—it’s the system that let a kid with no training, no supervision, and no accountability walk into a classroom with a loaded weapon. It wasn’t about “rights.” It was about loopholes. The system rewards compliance, not prevention. If we’re going to call this a “right,” we should also call it a responsibility that doesn’t exist for most kids. That’s not freedom—it’s failure."
\end{quote}

\subsubsection{Dataset details}

\paragraph{Research Questions $\mathbf{1}$ and $\mathbf{2}$.}
To answer research questions $1$ and $2$, we generate training and testing runs on random graphs for each LLM-topic pair.
In particular, for each LLM family, and each topic, we generate $24$ runs for training and $8$ runs for testing, where each run has a horizon of approximately $10$.

\paragraph{Research Question $\mathbf{3}$.}
When answering RQ3, we generate a fixed random graph for each LLM-topic pair. In particular, we
\begin{itemize}
    \item Generate a single Chung-Lu random graph for \textsc{Gemma3} agents discussing the vaccines topic, and we
    \item Generate a single community/stochastic block model graph for \textsc{Llama3.1} agents discussing the climate topic.
\end{itemize}
For \textsc{Gemma3} agents discussing the vaccines topic, we used $30$ training runs and $12$ testing runs.
For \textsc{Llama3.1} agents discussing the climate topic we used $28$ training runs and $12$ testing runs, where $2$ training runs were dropped as the agents halted discussion prematurely. 

\paragraph{Network size variations}
We detail additional experiments with larger network sizes in Appendix~\ref{sec:size_variation}. For these experiments, we generated $24$ training runs and omitted $2$ because they finished prematurely. We used $8$ runs for testing.

\subsection{Code and data availability}
Code, along with details to access data, can be found at 
\url{https://github.com/CalebP547/Characterizing_LLM_Opinion_Dynamics_Code}

We use the open-weight models \textsc{Llama3.1}, \textsc{Qwen3}, and \textsc{Gemma3}, the StanceAware-SBERT model, and standard optimization libraries (OSQP, CVXPY, SciPy) for their intended research purposes and in accordance with their terms of use. 

\onecolumn

\subsection{Details for fitting of opinion dynamics models}

\subsubsection{Tools used for fitting opinion dynamics models}

We used OSQP \cite{osqp} through the CVXPY \cite{agrawal2018rewriting,diamond2016cvxpy} interface to solve the quadratic programs required to fit the opinion dynamics model.

For the non-linear program required to fit a fully parameterized Friedkin-Johnsen model with bias and homophily effects, we used a SciPy \cite{2020SciPy-NMeth} implementation of a sequential least-squares quadratic program solver\footnote{\url{https://docs.scipy.org/doc/scipy/reference/optimize.minimize-slsqp.html}}.
To initialize the solver, we used
\begin{enumerate}
    \item $10$ random starts.
    \item The optimal solution found when fitting the model without homophily. 
    \item The optimal solution found when fitting the model without homophily, along with an initial $\gamma$ value of $0.1$. 
\end{enumerate}

Full details on the versions of Python libraries used can be found in the attached code base.

\newcommand{\op}[2]{z^{(#1)}(#2)}
\newcommand{\opdef}{\op{r}{t}}
\newcommand{\opdefi}[1]{z_{#1}^{(r)}(t)}
\newtheorem{lemma}{Lemma}
\newtheorem{remark}{Remark}

\subsubsection{Convex reformulations for opinion dynamics fitting problems} \label{apx:cvx_reformulation}

We reformulate some of the model-fitting problems to remove non-convex components. 
In particular, consider the following optimization problem in which we want to fit a Friedkin-Johnsen model that includes a bias component with a fully parameterized weight matrix $W$, where $\lambda_1$ corresponds to the initial opinion weight and $\lambda_2$ corresponds to the bias weight, and $H$ is a horizon.
\begin{equation}
\label{eq:fj_nonconvex_lsq}
\begin{aligned}
    \min_{W \in \mathbb{R}^{n\times n}, \lambda_1 \in \mathbb{R}, \lambda_2 \in \mathbb{R}, b \in \mathbb{R}} \ \ & 
    \sum_{r \in \mathcal{R}} \sum_{t = 1}^{H - 1}||\op{r}{t+1} - (\lambda_1 \op{r}{0} + \lambda_2 b\mathbf{1} + (1 - \lambda_1 - \lambda_2)W \opdef )||_2^2\\
    \text{s.t.} \ \ & \sum_{j} W_{ij} = 1, \quad \forall i \in [n],\\
     & W_{ij} \geq 0, \quad  \forall i,j \in [n], \\
     & b \in [-1,1], \\
     & \lambda_1, \lambda_2 \geq 0, \\
     & \lambda_1 + \lambda_2 \leq 1.
\end{aligned}
\end{equation}
The bilinear components of the above objective make the objective non-convex. However, we can reformulate this problem in a convex manner. 
In particular, we can solve the following least squares problem instead.
\begin{equation}
\label{eq:fj_convex_lsq}
\begin{aligned}
    \min_{\tilde{W} \in \mathbb{R}^{n\times n}, \lambda_1 \in \mathbb{R}, \lambda_2 \in \mathbb{R}, \tilde{b} \in \mathbb{R}} \ \ & 
    \sum_{r \in \mathcal{R}}\sum_{t = 1}^{H - 1}||\op{r}{t+1} - (\lambda_1 \op{r}{0} + \tilde{b}\mathbf{1} + \tilde{W} \opdef )||_2^2\\
    \text{s.t.} \ \ & \sum_{j} \tilde{W}_{ij} = (1 - \lambda_1 - \lambda_2), & \forall i \in [n],\\
     & \tilde{W}_{ij} \geq 0, & \forall i,j \in [n], \\
     & \tilde{b} \in [-\lambda_2, \lambda_2], \\
     & \lambda_1 \geq 0, \\
     & \lambda_2 \geq 0.
\end{aligned}
\end{equation}
The following Lemma codifies the equivalence between these problems.
\begin{lemma}
    \label{lem:opt_conversion}
    The following facts hold
    \begin{enumerate}
        \item Optimization problem \eqref{eq:fj_convex_lsq} is convex.
        \item Optimization problems \eqref{eq:fj_nonconvex_lsq} and \eqref{eq:fj_convex_lsq} have the same optimal value.
        \item Let $(\tilde{W}^*,\lambda_1^*,\lambda_2^*,\tilde{b}^*)$ be an optimal solution to \eqref{eq:fj_convex_lsq} such that $\lambda_1^* + \lambda_2^* < 1$ and $\lambda_2^* > 0$. Then $(\tilde{W}^*/(1 - \lambda_1^* - \lambda_2^*),\lambda_1^*,\lambda_2^*,\tilde{b}^*/\lambda_2^*)$ is an optimal solution to \eqref{eq:fj_nonconvex_lsq}.
    \end{enumerate}
\end{lemma}

\textit{Proof:}
1) The optimization problem \eqref{eq:fj_convex_lsq} is a quadratic program (QP) and is thus convex.

2) We prove this statement by constructing mappings between the feasible sets of these problems that preserve value.

Let $(W,\lambda_1,\lambda_2,b)$ be a feasible solution to \eqref{eq:fj_nonconvex_lsq}. 
The tuple $(\tilde{W}, \lambda_1,\lambda_2, \tilde{b}) = ((1 - \lambda_1 - \lambda_2)W, \lambda_1, \lambda_2, \lambda_2 b)$ is then a feasible solution to \eqref{eq:fj_convex_lsq}, and the objective value at this solution for problem \eqref{eq:fj_convex_lsq} is
\begin{multline}
    \sum_{r \in \mathcal{R}}\sum_{t = 1}^{H - 1}||\op{r}{t+1} - (\lambda_1 \op{r}{0} + \tilde{b}\mathbf{1} + \tilde{W} \opdef )||_2^2
    \\ = \sum_{r \in \mathcal{R}}\sum_{t = 1}^{H - 1}||\op{r}{t+1} - (\lambda_1 \op{r}{0} + \lambda_2 b\mathbf{1} + (1 - \lambda_1 - \lambda_2) W \opdef )||_2^2,
\end{multline}
which is exactly the objective value of the problem \eqref{eq:fj_convex_lsq} at $(W,\lambda_1,\lambda_2,b)$. 
Thus, the optimal value of \eqref{eq:fj_convex_lsq} is lower than that of \eqref{eq:fj_nonconvex_lsq}.

Similarly, let $(\tilde{W}, \lambda_1,\lambda_2, \tilde{b})$ be a solution to \eqref{eq:fj_convex_lsq}. 
If $(1 - \lambda_1 - \lambda_2) = 0$, then set $W = I_{n\times n}$, otherwise, set 
\begin{equation}
    W = \frac{\tilde{W}}{1 - \lambda_1 - \lambda_2}.
\end{equation}
Note that in both these cases for $(1 - \lambda_1 - \lambda_2)$, we have $(1 - \lambda_1 - \lambda_2) W = \tilde{W}$.
If $\lambda_2 > 0$, set $b = \tilde{b}/\lambda_2$, and if $\lambda_2 = 0$, set $b = 0$. 
Note that in both of these cases for $\lambda_2$, we have $b\lambda_2 = \tilde{b}$.
The objective value of \eqref{eq:fj_nonconvex_lsq} at the solution $(W,\lambda_1,\lambda_2,b)$ is
\begin{multline}
    \sum_{r \in \mathcal{R}}\sum_{t = 1}^{H - 1}||\op{r}{t+1} - (\lambda_1 \op{r}{0} + \lambda_2 b\mathbf{1} + (1 - \lambda_1 - \lambda_2) W \opdef )||_2^2
    \\ = \sum_{r \in \mathcal{R}}\sum_{t = 1}^{H - 1}||\op{r}{t+1} - (\lambda_1 \op{r}{0} + \tilde{b}\mathbf{1} + \tilde{W} \opdef )||_2^2,
\end{multline}
which is exactly the objective value of \eqref{eq:fj_convex_lsq} at $(\tilde{W},\lambda_1,\lambda_2,\tilde{b})$.
Thus, the optimal value of \eqref{eq:fj_nonconvex_lsq} is lower than that of \eqref{eq:fj_convex_lsq}.

We can now conclude that the two optimization problems have equal value. 

3) Let $(\tilde{W}^*,\lambda_1^*,\lambda_2^*,\tilde{b}^*)$ be an optimal solution to \eqref{eq:fj_convex_lsq} such that $\lambda_1^* + \lambda_2^* < 1$ and $\lambda_2^* > 0$, and let the objective value be $f^*_{\text{convex}}$.
By the above argument, we know that the objective value of the solution $(\tilde{W}^*/(1 - \lambda_1^* - \lambda_2^*),\lambda_1^*,\lambda_2^*,\tilde{b}^*/\lambda_2^*)$ in \eqref{eq:fj_nonconvex_lsq} is equal to $f^*_{\text{convex}}$. 
Additionally, letting $f^*_{\text{non-convex}}$ denote the optimal value for \eqref{eq:fj_convex_lsq}, and using the fact that $f^*_{\text{convex}} = f^*_{\text{non-convex}}$, we conclude that $(\tilde{W}^*/(1 - \lambda_1^* - \lambda_2^*),\lambda_1^*,\lambda_2^*,\tilde{b}^*/\lambda_2^*)$ is an optimal solution to \eqref{eq:fj_nonconvex_lsq}.
$\square$

\begin{remark}
    In the case where $\lambda_1^* + \lambda_2^* = 1$, the social interaction component is zero, and we can simply set $W$ to an arbitrary value to compute the corresponding opinion dynamics model from the solution to \eqref{eq:fj_convex_lsq}.
    Similarly, in the case where $\lambda_2^* = 0$, the bias effect is zero, and we can simply set $b$ to zero to compute the corresponding opinion dynamics model from the solution to \eqref{eq:fj_convex_lsq}.
\end{remark}

\begin{remark}
    We can easily add neighborhood constraints into this problem by constraining $W_{ij}$ to be zero when $j$ is not equal to $i$ and agent $j$ does not influence agent $i$.
\end{remark}

We apply a similar convexification procedure for fitting when we assume that $W$ is a function of the row-normalized adjacency matrix.
The naive formulation for fitting such a model is as follows, where $\lambda_1$ encodes the initial opinion weight, $\lambda_2$ encodes the bias weight, $\lambda_3$ encodes the self update weight and $\lambda_4$ encodes the neighborhood averaging weight.
\begin{equation}
\label{eq:nonconvex_adjacency_lsq}
\begin{aligned}
    \min_{\lambda_i \in \mathbb{R}, \tilde{b} \in \mathbb{R}} \ \ & \sum_{r \in \mathcal{R}} \sum_{t = 1}^{H - 1}
    ||\op{r}{t+1} - (\lambda_1 \op{r}{0} + \lambda_2 b \mathbf{1} + \lambda_3  \op{r}{t} + \lambda_4  Q(\opdef, \gamma) \opdef  )||_2^2\\
    \text{s.t.} \ \ & \lambda_i \geq 0, \quad \forall i \in \{1,2,3,4\},\\
     & \lambda_1 + \lambda_2 + \lambda_3 + \lambda_4 = 1, \\
     & b \in [-1,1].
\end{aligned}
\end{equation}
The matrix $Q(\opdef,\gamma)$ represents a generic row-stochastic matrix that includes the special cases of normalized adjacency matrices and normalized homophily kernels. In particular, when we do not model homophily, we set $Q(\opdef) = \bar{A}$. 
When we model homophily, and we fix $\gamma$, we set
\begin{equation}
     Q_{ij}(\opdef,\gamma) = H_{ij}(\bar{A}, \opdef, \gamma)=  \frac{\bar{A}_{ij} \exp(-\gamma |\opdefi{i} - \opdefi{j}| ) }{\sum_{j : (i,j) \in E} \bar{A}_{ij} \exp(-\gamma |\opdefi{i} - \opdefi{j}| )}.
\end{equation}
While the optimization problem \eqref{eq:nonconvex_adjacency_lsq} contains a bilinear component of the objective, 
we can use the same trick as above, and reformulate this problem as
\begin{equation}
\label{eq:convex_adjacency_lsq}
\begin{aligned}
    \min_{\lambda_i \in \mathbb{R}, \tilde{b} \in \mathbb{R}} \ \ & \sum_{r \in \mathcal{R}} \sum_{t = 1}^{H - 1}
    ||\op{r}{t+1} - (\lambda_1 \op{r}{0} + \tilde{b} \mathbf{1} + \lambda_3  \op{r}{t} + \lambda_4  Q(\opdef,\gamma) \opdef  )||_2^2\\
    \text{s.t.} \ \ & \lambda_i \geq 0, \quad \forall i \in \{1,3,4\},\\
     & \lambda_1 + \lambda_3 + \lambda_4 + \tilde{b} \leq 1, \\
     & \lambda_1 + \lambda_3 + \lambda_4 + (-\tilde{b}) \leq 1. \\
\end{aligned}
\end{equation}
The equivalence between \eqref{eq:nonconvex_adjacency_lsq} and \eqref{eq:convex_adjacency_lsq} follows the same logic as Lemma~\ref{lem:opt_conversion}.
In particular, given a solution $(\lambda_1^*,\lambda_3^*,\lambda_4^*,\tilde{b}^*)$ such that $(1 - \lambda_1^* - \lambda_3^* - \lambda_4^*) > 0$, we can recover $(\lambda_2^*,b^*)$ as
\begin{equation}
    b^* = \frac{\tilde{b}}{1 - \lambda_1^* - \lambda_3^* - \lambda_4^*}, \quad \lambda_2^* = 1 - \lambda_1^* - \lambda_3^* - \lambda_4^*.
\end{equation}
Note that in the case of $\lambda_1^* + \lambda_3^* + \lambda_4^* = 1$, the optimal solution does not include a bias term, i.e., $\lambda_2^* = 0$.

While the above reformulation assumes a fixed value for $\gamma$, we repeat this optimization for a range of $\gamma$ values in order to identify an optimal model as demonstrated in Figure~\ref{fig:homophily_sweep}.

\subsubsection{Formulation for fitting a full $W$ matrix for homophily}

Solving for the parameters of a fully parameterized Friedkin-Johnsen model with bias and homophily is a nonlinear program.
Recall the following definition for the row-normalized interaction matrix in models with homophily.
\begin{equation}
    H_{ij}(W, z, \gamma)=  \frac{W_{ij} \exp(-\gamma |z_i - z_j| ) }{\sum_{j : (i,j) \in E} W_{ij} \exp(-\gamma |z_i - z_j| )}.
\end{equation}
We solve the following nonlinear program to compute model parameters that match the data, where $\lambda_1$ encodes the initial opinion weight, $\lambda_2$ encodes the bias weight, $\lambda_3$ encodes the averaging weight. 
\begin{equation}
\label{eq:nonconvex solve}
\begin{aligned}
    \min_{W \in \mathbb{R}^{n\times n}, \lambda_1 \in \mathbb{R}, \lambda_3 \in \mathbb{R}, \tilde{b} \in \mathbb{R}, \gamma \in \mathbb{R}^+} \ \ & \sum_{r \in \mathcal{R}} \sum_{t = 1}^{H - 1}
    ||\op{r}{t+1} - ( \lambda_1 \op{r}{0} + \tilde{b} \mathbf{1} +  \lambda_3  H(W, \opdef, \gamma) \opdef)||_2^2\\
    \text{s.t.} \ \ & \lambda_i \geq 0, \quad \forall i \in \{1,3\},\\
     & \lambda_1 + \lambda_3  + \tilde{b} \leq 1, \\
     & \lambda_1 + \lambda_3  + (-\tilde{b}) \leq 1, \\
     & W_{ij} \geq 0, \quad \forall i,j, \\
     & W_{ij} = 0 \quad  \forall j : (i,j) \notin E \wedge \quad j \neq i. \\
\end{aligned}
\end{equation} 
Note that non-convexity appears via the normalization that appears in the definition of the matrix $H$.

\subsection{Extended experimental results}

We report the fit metrics for all models in Table~\ref{tab:model_comparison}. 
In particular, for each LLM family, topic, opinion dynamics model, run and time-step, we compute 
\begin{itemize}
    \item the distance between the predicted and observed mean opinions,
    \item the distance between the predicted and observed variances,
    \item the Wasserstein distance between predicted and observed opinion vectors.
\end{itemize}
For each LLM-topic pair, we truncate all runs to a common length, and for each run, we compute the run-maximum and run-integral values of the above metrics. In Table~\ref{tab:model_comparison}, we report the average of these metrics across the $8$ test runs.

\begin{table*}[t]
  \centering
  \scriptsize
  \caption{%
    Full model comparison across all LLMs and topics. Each row reports the out-of-sample predictive accuracy of one opinion dynamics model fitted to LLM interaction trajectories on a held-out test set. We abbreviate the models as follows:
    \textsc{DG} (DeGroot),
    \textsc{FJ} (Friedkin--Johnsen),
    \textsc{Homo} (Homophily weighting),
    \textit{Bias} (Opinion bias term). Combined labels indicate extensions of the base model. 
    We then use ``Bias+Init only'' and ``Bias only'' to denote the models discussed in Section~\ref{sec:methods_models}  that omit social interaction effects.
    $\bar{\varepsilon}_\mu$ and $\bar{\varepsilon}_\sigma$ denote mean and variance errors between predicted and observed opinion distributions, reported as both per-timestep maximum (\emph{Max}) and timestep-averaged (\emph{Integral}) quantities. Lower values indicate better fit. The best result in each group is \textbf{bolded}.
    }
  \label{tab:model_comparison}
  \renewcommand{\arraystretch}{0.85}
  \begin{tabularx}{\textwidth}{llrRRRRRR}
    \toprule
     &  &  & \multicolumn{2}{c}{Mean Error $\bar{\varepsilon}_\mu$} & \multicolumn{2}{c}{Variance Error $\bar{\varepsilon}_\sigma$} & \multicolumn{2}{c}{Wasserstein-1 Dist. $W_1$} \\
    \cmidrule(lr){4-5} \cmidrule(lr){6-7} \cmidrule(lr){8-9}
     &  &  & Max & Integral & Max & Integral & Max & Integral \\
    \midrule
    \multirow{24}{*}{\textsc{Gemma3}} & \multirow{8}{*}{Climate} & \textsc{DG} & 0.1304 & 0.5670 & 0.0674 & 0.3647 & 0.1776 & 1.0643 \\
     &  & \textsc{FJ} & 0.1299 & 0.5633 & 0.0662 & 0.3556 & $\mathbf{0.1719}$ & 1.0311 \\
     &  & \textsc{FJ}+\textsc{Bias} & 0.1152 & 0.4862 & 0.0756 & 0.4190 & 0.1871 & 1.1553 \\
     &  & \textsc{DG}+\textsc{Homo} & 0.1305 & 0.5763 & 0.0662 & 0.3608 & 0.1782 & 1.0484 \\
     &  & \textsc{FJ}+\textsc{Homo} & 0.1300 & 0.5721 & $\mathbf{0.0650}$ & $\mathbf{0.3519}$ & 0.1731 & $\mathbf{1.0170}$ \\
     &  & \textsc{FJ}+\textsc{Homo}+\textsc{Bias} & $\mathbf{0.1129}$ & $\mathbf{0.4782}$ & 0.0752 & 0.4178 & 0.1877 & 1.1541 \\
     &  & \textsc{Bias}+\textsc{Init} & 0.1168 & 0.5010 & 0.0753 & 0.4180 & 0.1882 & 1.1590 \\
     &  & \textsc{Bias Only} & 0.1176 & 0.5051 & 0.0757 & 0.4206 & 0.1920 & 1.1790 \\
    \cmidrule{2-9}
     & \multirow{8}{*}{Gun Control} & \textsc{DG} & 0.2424 & 1.3677 & 0.1700 & 0.9783 & 0.3504 & 2.1819 \\
     &  & \textsc{FJ} & 0.2421 & 1.3700 & $\mathbf{0.1673}$ & $\mathbf{0.9579}$ & 0.3448 & 2.1461 \\
     &  & \textsc{FJ}+\textsc{Bias} & $\mathbf{0.2046}$ & $\mathbf{1.0333}$ & 0.1795 & 1.0453 & 0.3403 & 2.1016 \\
     &  & \textsc{DG}+\textsc{Homo} & 0.2424 & 1.3677 & 0.1700 & 0.9783 & 0.3504 & 2.1819 \\
     &  & \textsc{FJ}+\textsc{Homo} & 0.2421 & 1.3700 & $\mathbf{0.1673}$ & $\mathbf{0.9579}$ & 0.3448 & 2.1461 \\
     &  & \textsc{FJ}+\textsc{Homo}+\textsc{Bias} & 0.2064 & 1.0516 & 0.1791 & 1.0433 & 0.3447 & 2.1262 \\
     &  & \textsc{Bias}+\textsc{Init} & 0.2063 & 1.0774 & 0.1753 & 1.0047 & $\mathbf{0.3382}$ & $\mathbf{2.0755}$ \\
     &  & \textsc{Bias Only} & 0.2073 & 1.0836 & 0.1762 & 1.0119 & 0.3422 & 2.0982 \\
    \cmidrule{2-9}
     & \multirow{8}{*}{Vaccines} & \textsc{DG} & 0.2157 & 1.3709 & $\mathbf{0.0335}$ & $\mathbf{0.1553}$ & 0.2234 & 1.5527 \\
     &  & \textsc{FJ} & 0.2157 & 1.3709 & $\mathbf{0.0335}$ & $\mathbf{0.1553}$ & 0.2234 & 1.5527 \\
     &  & \textsc{FJ}+\textsc{Bias} & $\mathbf{0.0504}$ & $\mathbf{0.2055}$ & 0.0376 & 0.1771 & $\mathbf{0.0961}$ & $\mathbf{0.5536}$ \\
     &  & \textsc{DG}+\textsc{Homo} & 0.1945 & 1.2156 & 0.0337 & 0.1572 & 0.2037 & 1.4294 \\
     &  & \textsc{FJ}+\textsc{Homo} & 0.1951 & 1.2183 & $\mathbf{0.0335}$ & 0.1559 & 0.2043 & 1.4297 \\
     &  & \textsc{FJ}+\textsc{Homo}+\textsc{Bias} & 0.0509 & 0.2056 & 0.0376 & 0.1774 & 0.0970 & 0.5593 \\
     &  & \textsc{Bias}+\textsc{Init} & $\mathbf{0.0504}$ & $\mathbf{0.2055}$ & 0.0376 & 0.1771 & $\mathbf{0.0961}$ & $\mathbf{0.5536}$ \\
     &  & \textsc{Bias Only} & 0.0509 & 0.2128 & 0.0382 & 0.1808 & 0.0995 & 0.5856 \\
    \midrule
    \multirow{24}{*}{\textsc{Llama3.1}} & \multirow{8}{*}{Climate} & \textsc{DG} & 0.2233 & 1.5328 & 0.0588 & 0.3254 & 0.2483 & 1.8536 \\
     &  & \textsc{FJ} & 0.2233 & 1.5328 & 0.0588 & 0.3252 & 0.2483 & 1.8532 \\
     &  & \textsc{FJ}+\textsc{Bias} & $\mathbf{0.0564}$ & $\mathbf{0.2862}$ & 0.0671 & 0.3716 & 0.1072 & 0.7442 \\
     &  & \textsc{DG}+\textsc{Homo} & 0.1954 & 1.3051 & 0.0571 & 0.3199 & 0.2242 & 1.6514 \\
     &  & \textsc{FJ}+\textsc{Homo} & 0.1957 & 1.3061 & $\mathbf{0.0567}$ & $\mathbf{0.3167}$ & 0.2235 & 1.6459 \\
     &  & \textsc{FJ}+\textsc{Homo}+\textsc{Bias} & $\mathbf{0.0564}$ & $\mathbf{0.2862}$ & 0.0671 & 0.3716 & 0.1072 & 0.7442 \\
     &  & \textsc{Bias}+\textsc{Init} & 0.0570 & 0.2968 & 0.0668 & 0.3705 & $\mathbf{0.1062}$ & $\mathbf{0.7410}$ \\
     &  & \textsc{Bias Only} & 0.0590 & 0.3103 & 0.0684 & 0.3814 & 0.1145 & 0.8170 \\
    \cmidrule{2-9}
     & \multirow{8}{*}{Gun Control} & \textsc{DG} & 0.4249 & 2.8675 & 0.1351 & 0.6369 & 0.4475 & 3.2445 \\
     &  & \textsc{FJ} & 0.4252 & 2.8695 & 0.1341 & 0.6291 & 0.4470 & 3.2381 \\
     &  & \textsc{FJ}+\textsc{Bias} & $\mathbf{0.1052}$ & $\mathbf{0.5118}$ & 0.1395 & 0.6714 & $\mathbf{0.1807}$ & $\mathbf{1.2725}$ \\
     &  & \textsc{DG}+\textsc{Homo} & 0.4015 & 2.6864 & 0.1338 & 0.6300 & 0.4269 & 3.0859 \\
     &  & \textsc{FJ}+\textsc{Homo} & 0.4019 & 2.6851 & $\mathbf{0.1316}$ & $\mathbf{0.6134}$ & 0.4257 & 3.0676 \\
     &  & \textsc{FJ}+\textsc{Homo}+\textsc{Bias} & $\mathbf{0.1052}$ & $\mathbf{0.5118}$ & 0.1395 & 0.6714 & $\mathbf{0.1807}$ & $\mathbf{1.2725}$ \\
     &  & \textsc{Bias}+\textsc{Init} & 0.1056 & 0.5174 & 0.1394 & 0.6703 & 0.1808 & 1.2741 \\
     &  & \textsc{Bias Only} & 0.1134 & 0.5830 & 0.1468 & 0.7319 & 0.2092 & 1.4948 \\
    \cmidrule{2-9}
     & \multirow{8}{*}{Vaccines} & \textsc{DG} & 0.2011 & 1.5576 & 0.0361 & 0.1903 & 0.2202 & 1.7830 \\
     &  & \textsc{FJ} & 0.2016 & 1.5604 & 0.0359 & $\mathbf{0.1882}$ & 0.2199 & 1.7806 \\
     &  & \textsc{FJ}+\textsc{Bias} & 0.0440 & 0.1943 & 0.0383 & 0.2001 & 0.0920 & 0.6070 \\
     &  & \textsc{DG}+\textsc{Homo} & 0.1836 & 1.4028 & 0.0361 & 0.1968 & 0.2065 & 1.6726 \\
     &  & \textsc{FJ}+\textsc{Homo} & 0.1855 & 1.4112 & $\mathbf{0.0356}$ & 0.1923 & 0.2063 & 1.6664 \\
     &  & \textsc{FJ}+\textsc{Homo}+\textsc{Bias} & 0.0430 & 0.1923 & 0.0381 & 0.1995 & 0.0916 & 0.6135 \\
     &  & \textsc{Bias}+\textsc{Init} & $\mathbf{0.0426}$ & $\mathbf{0.1901}$ & 0.0381 & 0.1990 & $\mathbf{0.0908}$ & $\mathbf{0.6056}$ \\
     &  & \textsc{Bias Only} & 0.0469 & 0.2260 & 0.0395 & 0.2129 & 0.0999 & 0.7085 \\
    \midrule
    \multirow{24}{*}{\textsc{Qwen3}} & \multirow{8}{*}{Climate} & \textsc{DG} & 0.1928 & 1.2660 & $\mathbf{0.0461}$ & $\mathbf{0.2001}$ & 0.2039 & 1.4871 \\
     &  & \textsc{FJ} & 0.1928 & 1.2660 & $\mathbf{0.0461}$ & $\mathbf{0.2001}$ & 0.2039 & 1.4871 \\
     &  & \textsc{FJ}+\textsc{Bias} & 0.0975 & 0.5588 & 0.0524 & 0.2339 & 0.1254 & $\mathbf{0.9049}$ \\
     &  & \textsc{DG}+\textsc{Homo} & 0.1748 & 1.1247 & 0.0476 & 0.2083 & 0.1850 & 1.3536 \\
     &  & \textsc{FJ}+\textsc{Homo} & 0.1748 & 1.1247 & 0.0476 & 0.2083 & 0.1850 & 1.3536 \\
     &  & \textsc{FJ}+\textsc{Homo}+\textsc{Bias} & $\mathbf{0.0956}$ & $\mathbf{0.5477}$ & 0.0527 & 0.2347 & $\mathbf{0.1252}$ & 0.9092 \\
     &  & \textsc{Bias}+\textsc{Init} & 0.0984 & 0.5761 & 0.0517 & 0.2302 & 0.1254 & 0.9088 \\
     &  & \textsc{Bias Only} & 0.0987 & 0.5779 & 0.0518 & 0.2314 & 0.1262 & 0.9143 \\
    \cmidrule{2-9}
     & \multirow{8}{*}{Gun Control} & \textsc{DG} & 0.3477 & 1.5007 & 0.1326 & 0.4566 & 0.3837 & 1.8726 \\
     &  & \textsc{FJ} & 0.3477 & 1.5007 & 0.1326 & 0.4566 & 0.3837 & 1.8726 \\
     &  & \textsc{FJ}+\textsc{Bias} & 0.1303 & 0.5049 & 0.1529 & 0.5272 & 0.2226 & 1.0783 \\
     &  & \textsc{DG}+\textsc{Homo} & 0.3368 & 1.4350 & $\mathbf{0.1185}$ & $\mathbf{0.4191}$ & 0.3697 & 1.7860 \\
     &  & \textsc{FJ}+\textsc{Homo} & 0.3368 & 1.4350 & $\mathbf{0.1185}$ & $\mathbf{0.4191}$ & 0.3697 & 1.7860 \\
     &  & \textsc{FJ}+\textsc{Homo}+\textsc{Bias} & 0.1317 & 0.5125 & 0.1506 & 0.5223 & 0.2221 & 1.0834 \\
     &  & \textsc{Bias}+\textsc{Init} & $\mathbf{0.1298}$ & $\mathbf{0.5030}$ & 0.1515 & 0.5226 & $\mathbf{0.2178}$ & $\mathbf{1.0719}$ \\
     &  & \textsc{Bias Only} & $\mathbf{0.1298}$ & $\mathbf{0.5030}$ & 0.1515 & 0.5226 & $\mathbf{0.2178}$ & $\mathbf{1.0719}$ \\
    \cmidrule{2-9}
     & \multirow{8}{*}{Vaccines} & \textsc{DG} & 0.1445 & 0.8508 & $\mathbf{0.0265}$ & $\mathbf{0.0997}$ & 0.1541 & 1.0094 \\
     &  & \textsc{FJ} & 0.1445 & 0.8508 & $\mathbf{0.0265}$ & $\mathbf{0.0997}$ & 0.1541 & 1.0094 \\
     &  & \textsc{FJ}+\textsc{Bias} & $\mathbf{0.0702}$ & $\mathbf{0.3377}$ & 0.0304 & 0.1210 & 0.0926 & $\mathbf{0.5540}$ \\
     &  & \textsc{DG}+\textsc{Homo} & 0.1290 & 0.7561 & 0.0271 & 0.1045 & 0.1452 & 0.9408 \\
     &  & \textsc{FJ}+\textsc{Homo} & 0.1290 & 0.7561 & 0.0271 & 0.1045 & 0.1452 & 0.9408 \\
     &  & \textsc{FJ}+\textsc{Homo}+\textsc{Bias} & 0.0713 & 0.3427 & 0.0304 & 0.1216 & 0.0942 & 0.5683 \\
     &  & \textsc{Bias}+\textsc{Init} & 0.0718 & 0.3445 & 0.0299 & 0.1189 & $\mathbf{0.0917}$ & 0.5573 \\
     &  & \textsc{Bias Only} & 0.0718 & 0.3445 & 0.0299 & 0.1189 & $\mathbf{0.0917}$ & 0.5573 \\
    \bottomrule
  \end{tabularx}
\end{table*}

\FloatBarrier

\subsection{Sensitivity of the mean-square-error to fitted parameters}
\label{app:hessian}
We evaluate the sensitivity of the fitted parameters by evaluating the Hessian of the MSE at the optimal parameters we find, and we report these values in Table~\ref{tab:lambda_sensitivity}.
Let $g(\lambda_{\mathrm{self}}, \lambda_{\mathrm{neigh}}, \lambda_{0}, \lambda_{b}, b, \gamma)$ denote the MSE on the training data evaluated at a given set of model parameters.
We evaluate the sensitivity of $g$ to perturbations in each $\lambda$ parameter in $\{\lambda_{\mathrm{neigh}}, \lambda_{0}, \lambda_{b}\}$, as well as in the bias term $b$.
When evaluating sensitivity with respect to a $\lambda$ parameter, we decrease $\lambda_{\mathrm{self}}$ by the same amount to compensate for the increase in the perturbed parameter, thereby ensuring that the weight values continue to sum to $1$. 
For example, when evaluating sensitivity with respect to $\lambda_{\mathrm{neigh}}$, we evaluate the Hessian of the function
\begin{equation}
    p \mapsto g(\lambda_{\mathrm{self}} - p, \lambda_{\mathrm{neigh}} + p, \lambda_{0}, \lambda_{b}, b, \gamma)
\end{equation}
We omit sensitivity values for $\lambda_{\mathrm{self}}$, as this parameter is determined by the values of the remaining $\lambda$ parameters.

\begin{table*}[h]
  \centering
  \small
  \caption{%
    Sensitivity of MSE to perturbations in the fitted parameters of the \textsc{FJ}+\textsc{Homo}+\textit{Bias} model, across LLMs and topics. 
    ``Sensitivity'' is the raw, unnormalized sensitivity, while ``Normalized
    Sensitivity'' represents the same value rescaled by the MSE.
    Entries marked by ``b'' correspond to fitted parameters that were close to constraint boundaries.
    In these cases, we do not report second-order information as the Hessian, in this case, does not encode sensitivity.
  }
  \label{tab:lambda_sensitivity}
  \begin{tabularx}{\textwidth}{llSSSSSSSSS}
    \toprule
    \multirow{2}{*}{LLM} & \multirow{2}{*}{Topic} & \multirow{2}{*}{MSE \ } & \multicolumn{4}{c}{Sensitivity} & \multicolumn{4}{c}{Normalized Sensitivity} \\
    \cmidrule(lr){4-7} \cmidrule(lr){8-11}
    & & & $\lambda_{\mathrm{neigh}}$ & $\lambda_{0}$ & $\lambda_{b}$ & $b$ & $\lambda_{\mathrm{neigh}}$ & $\lambda_{0}$ & $\lambda_{b}$ & $b$ \\
    \midrule
    \multirow{3}{*}{\textsc{Gemma3}} & Climate & $0.0300$ & $0.1105$ & $0.2737$ & $0.1393$ & $0.1891$ & $3.6799$ & $9.1143$ & $4.6375$ & $6.2979$ \\
     & Gun Control & $0.0939$ & $0.4019$ & b & $0.4711$ & $0.0588$ & $4.2827$ & b & $5.0194$ & $0.6269$ \\
     & Vaccines & $0.0128$ & $0.0385$ & $0.2089$ & $0.0860$ & $0.2748$ & $2.9990$ & $16.2669$ & $6.6954$ & $21.4038$ \\
    \midrule
    \multirow{3}{*}{\textsc{Llama3.1}} & Climate & $0.0196$ & $0.1382$ & $0.3608$ & $0.1569$ & $0.1904$ & $7.0341$ & $18.3597$ & $7.9874$ & $9.6904$ \\
     & Gun Control & $0.0593$ & $0.4155$ & $0.8783$ & $0.4913$ & $0.1671$ & $7.0070$ & $14.8122$ & $8.2863$ & $2.8189$ \\
     & Vaccines & $0.0116$ & $0.0286$ & $0.1764$ & $0.1017$ & $0.1365$ & $2.4608$ & $15.1949$ & $8.7655$ & $11.7599$ \\
    \midrule
    \multirow{3}{*}{\textsc{Qwen3}} & Climate & $0.0168$ & $0.1008$ & $0.2597$ & $0.1447$ & $0.0796$ & $5.9913$ & $15.4396$ & $8.6007$ & $4.7354$ \\
     & Gun Control & $0.0522$ & $0.3010$ & b & $0.4751$ & $0.1447$ & $5.7624$ & b & $9.0959$ & $2.7711$ \\
     & Vaccines & $0.0093$ & $0.0389$ & b & $0.0712$ & $0.0964$ & $4.1636$ & b & $7.6251$ & $10.3229$ \\
    \bottomrule
  \end{tabularx}
\end{table*}

\FloatBarrier

\newpage
\subsection{Stance model validation}
\label{sec:stance_model_validation}

\subsubsection{Validation against human ratings}

To evaluate the quality of the stance modeling, we compare the opinion measurements generated by the stance-aware sentence transformer with human ratings in a small-scale experiment. 
In particular, we randomly sample $100$ messages from approximately the first $15\%$ of the training runs for conversations among \textsc{Llama 3.1} agents on the climate topic. 
We compile these messages into a spreadsheet for human evaluation, where each human rates each message on the five-point scale shown in Table~\ref{tab:climate_rating_scale}. 
Each human evaluator was an author of the paper.

\begin{table}[h]
\centering
\begin{tabular}{|c|p{11cm}|}
\hline
\textbf{Rating} & \textbf{Statement} \\
\hline
1 & The author of this message \textbf{strongly disagrees} with the statement \textit{"Climate Change is caused by human activities"} \\
\hline
2 & The author of this message \textbf{disagrees} with the statement \textit{"Climate Change is caused by human activities"} \\
\hline
3 & The author of this message is \textbf{neutral} with regard to the statement \textit{"Climate Change is caused by human activities"} \\
\hline
4 & The author of this message \textbf{agrees} with the statement \textit{"Climate Change is caused by human activities"} \\
\hline
5 & The author of this message \textbf{strongly agrees} with the statement \textit{"Climate Change is caused by human activities"} \\
\hline
\end{tabular}
\caption{Rating scale for stances on climate change.}
\label{tab:climate_rating_scale}
\end{table}

The human ratings and stance scores exhibit moderate correlation. 
We evaluate the correspondence between the human ratings and stance scores using the polyserial correlation \cite{olsson1982polyserial}, computed using the R implementation\footnote{\url{https://search.r-project.org/CRAN/refmans/polycor/html/polyserial.html}}. We present the results in Figure~\ref{fig:human_stance_model_correlation}.
Qualitatively, we find that the stance score accurately identifies messages that humans rate as $5$, while the stance model's ratings are noisier for messages that express skepticism toward the statement, ``Climate change is caused by human activities''.
\begin{figure*}[h]
    \centering
    \includegraphics[width=0.9\linewidth]{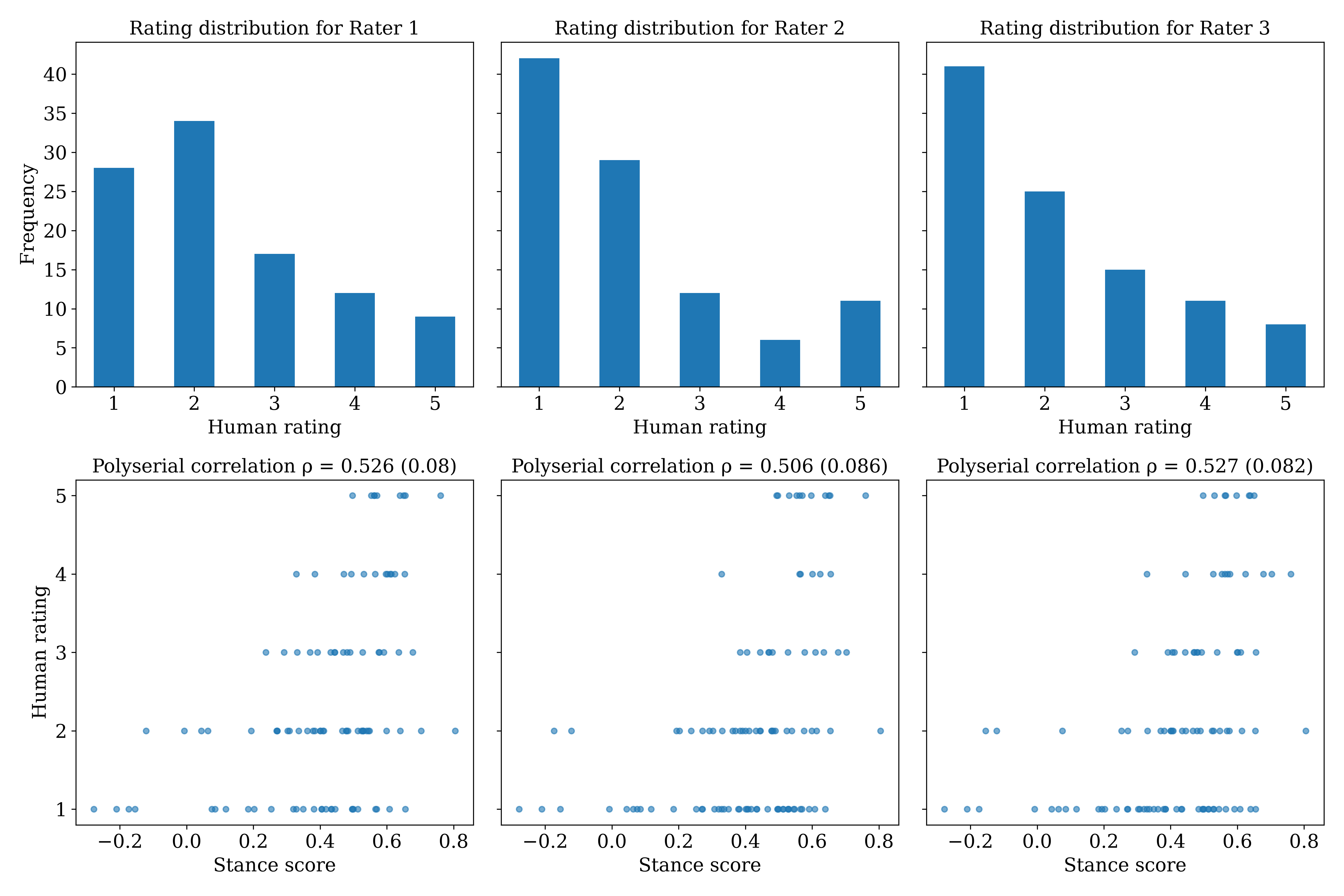}
    \caption{Correlation between human and stance-model ratings of LLM-generated posts.}
    \label{fig:human_stance_model_correlation}
\end{figure*}

Overall, these correlations provide evidence that the embedding-based stance scores are broadly consistent with human judgments of expressed stance.

\subsubsection{Sensitivity of stances to prompt variations}

We additionally examine the sensitivity of the opinion measurements to variations in the texts $a_{\text{pro}}$ and $a_{\text{anti}}$. 
Given a base stance \emph{base}, we define two natural pro and anti texts as follows.
\begin{equation}
    a_{\text{pro}}(\emph{base}) = \text{\emph{base}} \quad \quad a_{\text{anti}}(\emph{base}) = \text{It is not true that \emph{base}}.
\end{equation}
Changing the \emph{base} text leads to slight variations in the resulting opinion measurements. We therefore examine the correlation between opinion measurements produced by measurement schemes corresponding to different choices of \emph{base}. Specifically, we estimate the opinions expressed in messages generated by agents during discussion.
In Figures~\ref{fig:stance_sens_gemma_gc},\ref{fig:stance_sens_llama_clim}, and \ref{fig:stance_sens_qwen_vacc}, we plot variation in stance scores resulting from different choices of \emph{base} text for
\begin{itemize}
    \item \textsc{Gemma3} agents discussing the gun-control topic,
    \item \textsc{LLama3.1} agents discussing the climate topic, and
    \item \textsc{Qwen3} agents discussing the vaccine topic.
\end{itemize}

The generated opinion measurements are largely insensitive to the underlying texts used. Across the \textsc{Gemma3}-gun-control and \textsc{Llama3.1}-climate-change texts, we find that the opinion measurements are robust to variations in phrasing and tense, with the primary effect being a change in the scale of the measurements. 
Although the correlation is weaker when we modify the opinion-measurement pipeline for messages generated by \textsc{Qwen3} agents discussing vaccines, we still observe a strong positive correlation.
\begin{figure*}[h]
    \centering
    \includegraphics[width=\linewidth]{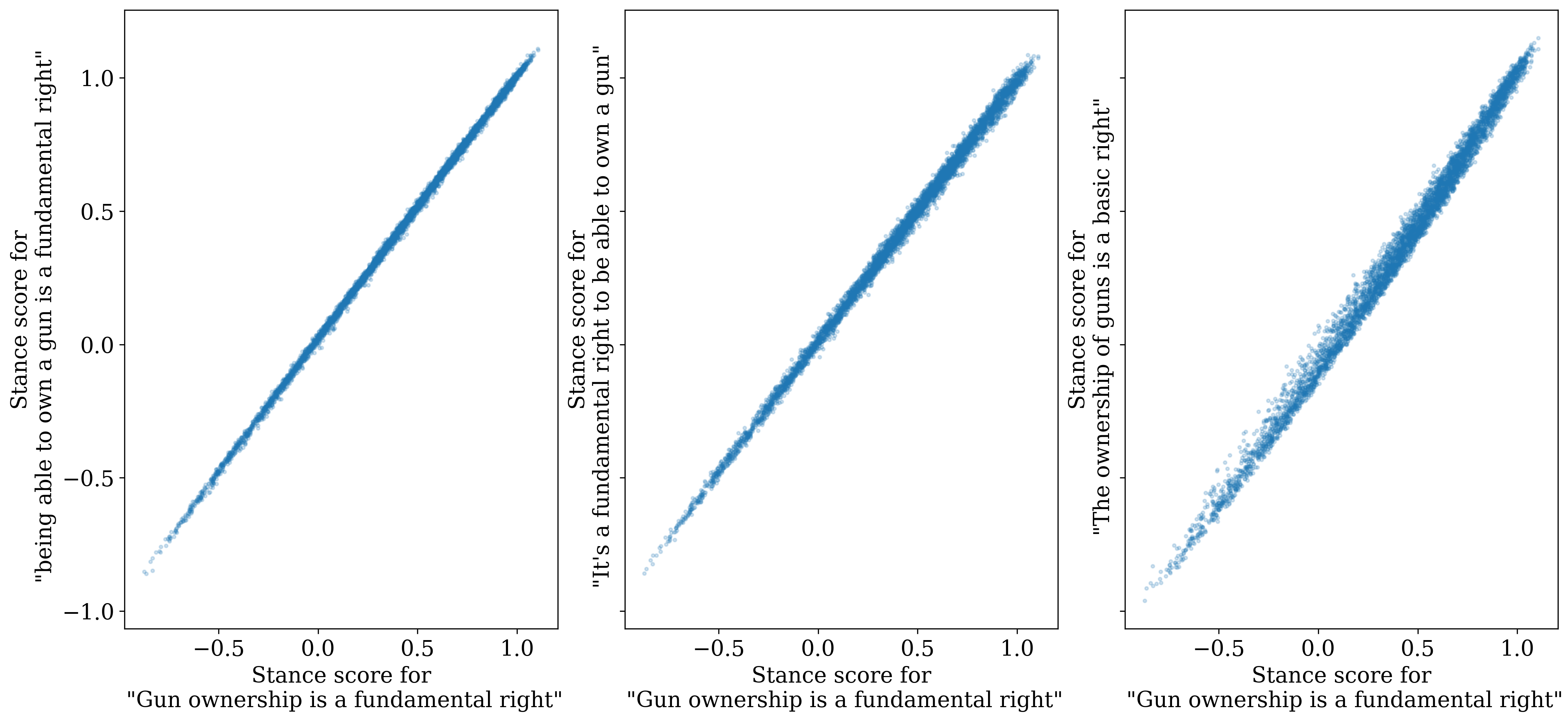}
    \caption{Relationship between opinion measurements generated by various \emph{base} texts when \textsc{Gemma3} agents discuss gun ownership.
    }
    \label{fig:stance_sens_gemma_gc}
\end{figure*}
\begin{figure*}[h]
    \centering
    \includegraphics[width=\linewidth]{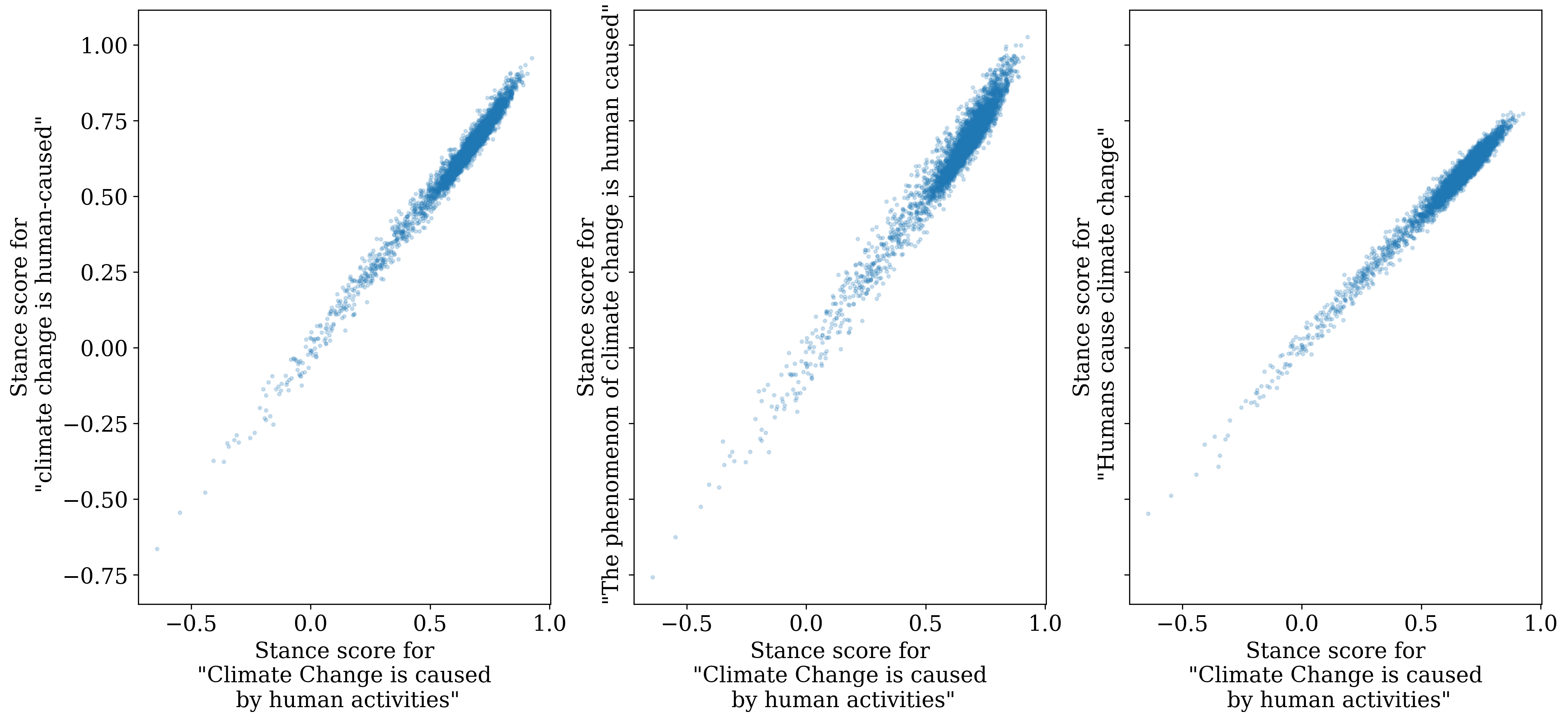}
    \caption{Relationship between opinion measurements generated by various \emph{base} texts when \textsc{Llama3.1} agents discuss climate change.}
    \label{fig:stance_sens_llama_clim}
\end{figure*}
\begin{figure*}[h]
    \centering
    \includegraphics[width=\linewidth]{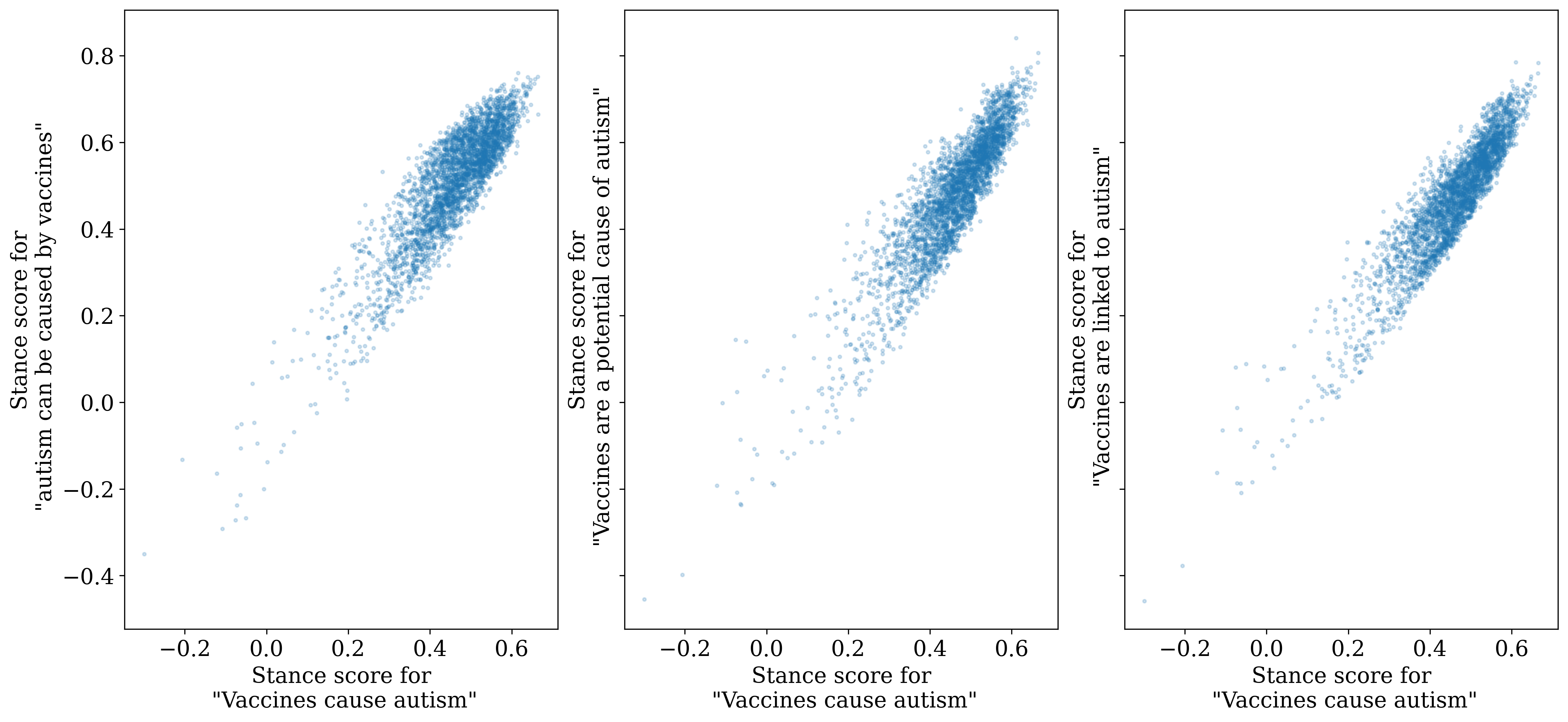}
    \caption{Relationship between opinion measurements generated by various \emph{base} texts when \textsc{Qwen3} agents discuss vaccines.}
    \label{fig:stance_sens_qwen_vacc}
\end{figure*}

\FloatBarrier

\subsection{Experiments for $60$ agents}
\label{sec:size_variation}

To evaluate the robustness of our results to network size, we run additional simulations with larger networks. In particular, we simulate $60$ \textsc{Llama3.1} agents discussing the vaccines topic, with the results presented in Table~\ref{tab:model_comparison_n60} and Figure~\ref{fig:rq4_60_agents}. We use $24$ training runs, of which $2$ are omitted because they terminated prematurely, and use $8$ additional runs for testing.

We find that the conclusions in Section~\ref{sec:experiments_RQ1} are robust to network size. 
In particular, the bias term is a critical component for ensuring predictive accuracy when modeling the opinion dynamics of networked LLMs. The Friedkin-Johnsen model, augmented with bias and homophily, performs best in terms of mean error, while the model with only bias and initial opinion stubbornness performs best in terms of Wasserstein error. Meanwhile, variance error remains relatively unchanged across the different opinion dynamics models

\begin{table*}[h]
  \centering
  \footnotesize
  \caption{Model comparison for $n=60$ \textsc{Llama3.1} agents discussing the vaccines topic.
  Each row reports the out-of-sample predictive accuracy of one opinion dynamics model fitted to LLM interaction trajectories on a held-out test set. We abbreviate the models as follows:
    \textsc{DG} (DeGroot),
    \textsc{FJ} (Friedkin--Johnsen),
    \textsc{Homo} (Homophily weighting),
    \textit{Bias} (Opinion bias term). Combined labels indicate extensions of the base model. 
    We then use ``Bias+Init only'' and ``Bias only'' to denote the models discussed in Section~\ref{sec:methods_models}  that omit social interaction effects.
    $\bar{\varepsilon}_\mu$ and $\bar{\varepsilon}_\sigma$ denote mean and variance errors between predicted and observed opinion distributions, reported as both per-timestep maximum (\emph{Max}) and timestep-averaged (\emph{Integral}) quantities. Lower values indicate better fit. The best result in each group is \textbf{bolded}.}
  \label{tab:model_comparison_n60}
  \begin{tabularx}{\textwidth}{lRRRRRR}
    \toprule
     & \multicolumn{2}{c}{Mean Error $\bar{\varepsilon}_\mu$} & \multicolumn{2}{c}{Variance Error $\bar{\varepsilon}_\sigma$} & \multicolumn{2}{c}{Wasserstein-1 Dist. $W_1$} \\
    \cmidrule(lr){2-3} \cmidrule(lr){4-5} \cmidrule(lr){6-7}
    Model & Max & Integral & Max & Integral & Max & Integral \\
    \midrule
    \textsc{DG} & 0.2281 & 1.2852 & $\mathbf{0.0489}$ & $\mathbf{0.2935}$ & 0.2498 & 1.5217 \\
    \textsc{FJ} & 0.2281 & 1.2852 & $\mathbf{0.0489}$ & $\mathbf{0.2935}$ & 0.2498 & 1.5217 \\
    \textsc{FJ}+\textsc{Bias} & 0.0352 & 0.1513 & 0.0512 & 0.3067 & 0.0994 & 0.6484 \\
    \textsc{DG}+\textsc{Homo} & 0.2111 & 1.1751 & 0.0514 & 0.3075 & 0.2379 & 1.4497 \\
    \textsc{FJ}+\textsc{Homo} & 0.2111 & 1.1751 & 0.0514 & 0.3075 & 0.2379 & 1.4497 \\
    \textsc{FJ}+\textsc{Homo}+\textsc{Bias} & $\mathbf{0.0351}$ & $\mathbf{0.1502}$ & 0.0514 & 0.3080 & 0.1005 & 0.6575 \\
    \textsc{Bias}+\textsc{Init} & 0.0357 & 0.1579 & 0.0508 & 0.3046 & $\mathbf{0.0986}$ & $\mathbf{0.6483}$ \\
    \textsc{Bias} & 0.0392 & 0.1767 & 0.0532 & 0.3209 & 0.1084 & 0.7094 \\
    \bottomrule
  \end{tabularx}
\end{table*}

\begin{figure*}[h]
    \centering
    \includegraphics[width=\linewidth]{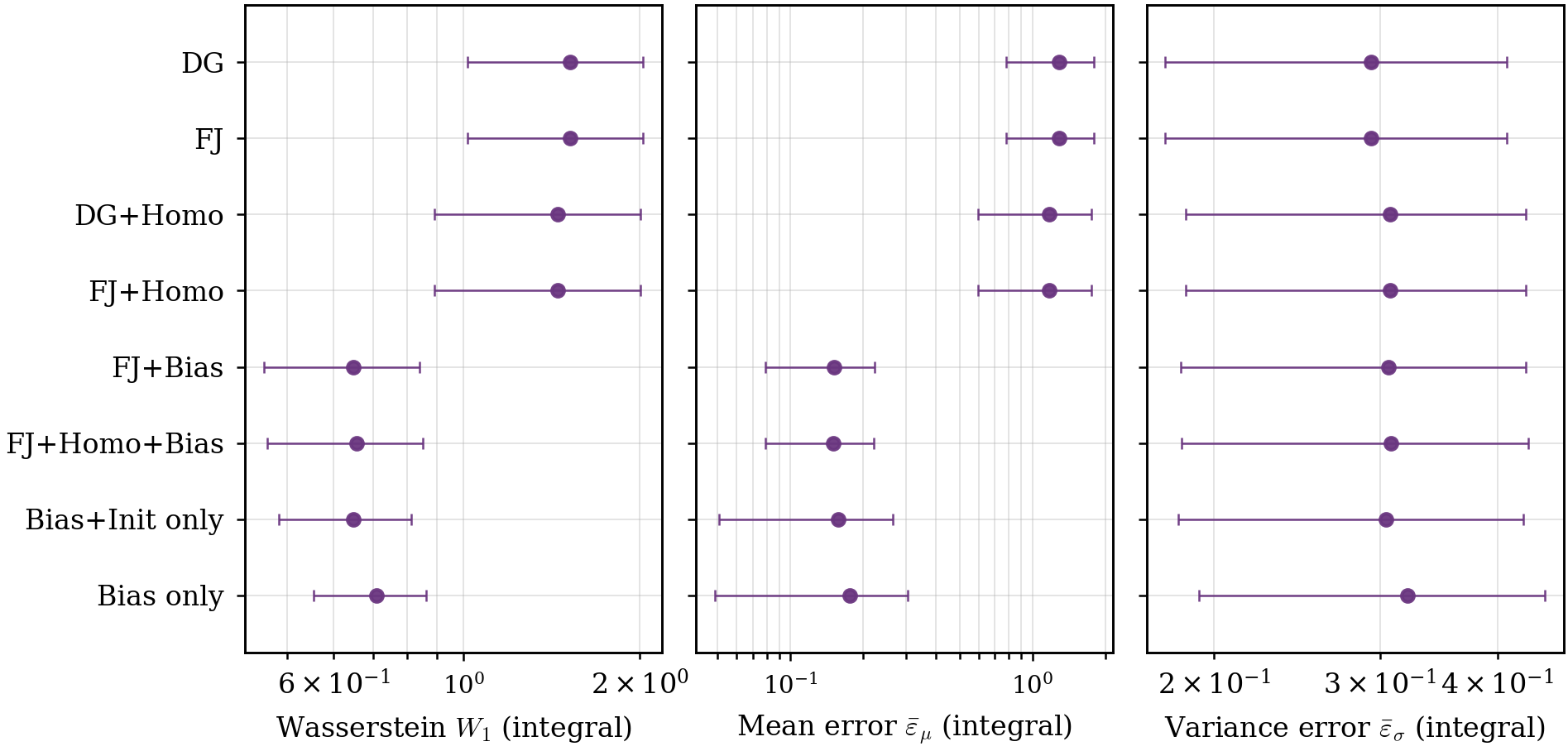}
    \caption{Out-of-sample prediction error for $n=60$ \textsc{Llama3.1} agents discussing the vaccines topic reported on three metrics. Each point shows the mean error across held-out runs.  Whiskers show one standard deviation. We abbreviate the models as follows:
    DG (DeGroot),
    FJ (Friedkin--Johnsen),
    Homo (Homophily weighting),
    Bias (Opinion bias term).
    We additionally use "Bias + Init" to refer to a model with only initial opinion and bias terms, along with a self-weight term.}
    \label{fig:rq4_60_agents}
\end{figure*}

\color{black}

\end{document}